\documentclass[9pt,twocolumn,twoside]{pnas-new}

\templatetype{pnasresearcharticle} 

\usepackage{color,soul}

\usepackage[T1]{fontenc}			
\usepackage[sc,osf]{mathpazo}   	

\usepackage{amsmath}  			
\usepackage{amsfonts}  			
\usepackage{graphicx}   			
\usepackage{mathrsfs, amsthm, amssymb}
\usepackage{bbm, bm}

\usepackage{nicefrac}    			

\usepackage{fourier-orns}  		

\usepackage[absolute]{textpos}

\usepackage[braket, qm]{qcircuit}	

\newtheorem{theorem}{Theorem}
\newtheorem{definition}{Definition}

\newtheorem{lemma}{Lemma}

\newtheorem{remark}{Remark}

\title{Violations of locality and free choice are equivalent resources in Bell experiments}

\author[a,b,1]{Pawel Blasiak}
\author[b]{Emmanuel M. Pothos} 
\author[b]{James M. Yearsley}
\author[c]{Christoph Gallus}
\author[a]{Ewa Borsuk}

\affil[a]{Institute of Nuclear Physics Polish Academy of Sciences, PL-31342 Krakow, Poland}
\affil[b]{City, University of London, London EC1V 0HB, United Kingdom}
\affil[c]{Technische Hochschule Mittelhessen, D-35390 Gießen, Germany}

\leadauthor{Blasiak} 

\significancestatement{Faced with a violation of Bell inequalities, a committed realist might pursue an explanation of the observed correlations on the basis of violations of the \textit{locality} or \textit{free} \textit{choice} (sometimes called \textit{measurement} \textit{independence}) assumptions. The question of whether it is better to abandon (partially or completely) locality or free choice has been strongly debated since the inception of Bell inequalities, with ardent supporters on either side. We offer a comprehensive treatment that allows a comparison of both assumptions on an equal footing, demonstrating a deep interchangeability. This advances the foundational debate and provides quantitative answers regarding the weight of each assumption for causal (or realist) explanations of observed correlations.}

\authorcontributions{P.B. designed and performed the research; P.B, E.M.P, J.M.Y. C.G. and E.B. discussed the results and wrote the paper.}

\authordeclaration{The authors declare no competing interests.}
\correspondingauthor{\textsuperscript{1}To whom correspondence should be addressed. E-mail: pawel.blasiak@ifj.edu.pl}


\begin{abstract}
Bell inequalities rest on three fundamental assumptions: \textit{realism}, \textit{locality}, and \textit{free choice}, which lead to nontrivial constraints on correlations in very simple experiments. If we retain realism, then violation of the inequalities implies that at least one of the remaining two assumptions must fail, which can have profound consequences for the causal explanation of the experiment. We investigate the extent to which a given assumption needs to be relaxed for the other to hold at all costs, based on the observation that a violation need not occur on every experimental trial, even when describing correlations violating Bell inequalities. How often this needs to be the case determines the degree of, respectively, \textit{locality} or \textit{free choice} in the observed experimental behavior. Despite their disparate character, we show that both assumptions are \textit{equally costly}. Namely, the resources required to explain the experimental statistics (measured by the frequency of causal interventions of either sort) are exactly the same. Furthermore, we compute such defined measures of locality and free choice for any \textit{nonsignaling} statistics in a Bell experiment with binary settings, showing that it is directly related to the amount of violation of the so-called Clauser-Horne-Shimony-Holt inequalities. This result is theory independent as it refers directly to the experimental statistics. Additionally, we show how the local fraction results for quantum-mechanical frameworks with infinite number of settings translate into analogous statements for the measure of free choice we introduce. Thus, concerning statistics, causal explanations resorting to either locality or free choice violations are \textit{fully} interchangeable.
\end{abstract}

\doi{\href{https://doi.org/10.1073/pnas.2020569118}{www.pnas.org/cgi/doi/10.1073/pnas.2020569118}}

\begin{document}

\maketitle
\thispagestyle{firststyle}
\ifthenelse{\boolean{shortarticle}}{\ifthenelse{\boolean{singlecolumn}}{\abscontentformatted}{\abscontent}}{}

\begin{textblock}{10}(7.8,0.5) For published version see \href{https://doi.org/10.1073/pnas.2020569118}{\bf PNAS 118 (17) e2020569118 (2021)}
\end{textblock}

\begin{quote}
\begin{flushright}
\textit{"I would rather discover one true cause\\than gain the kingdom of Persia."}

\textit{-- Democritus (c. 460-370 BC)}
\end{flushright}
\end{quote}

\dropcap{T}he study of experimental correlations provides a window into the underlying causal mechanisms, even when their exact nature remains obscured. In his seminal works~\cite{Be93,Me93,Wi14a,BrCaPiScWe14,Sc19}, John Bell showed that seemingly innocuous assumptions about the structure of causal relationships leave a mark on the observed statistics. The first assumption, called \textit{realism} (or counterfactual definiteness), presents the worldview in which physical objects and their properties exist, whether they are observed or not. Note that realism allows a standard notion of \textit{causality}~\cite{Be87a,Pe09}, which in turn provides us with the language to express the remaining two assumptions. The \textit{locality} assumption is a statement that physical (or causal) influences propagate in accord with the spatio-temporal structure of events (i.e., neither backward in time nor instantaneous causation). The \textit{free} \textit{choice} assumption asserts that the choice of measurement settings can be made independently from anything  in the (causal) past. These three assumptions are enough to derive testable constraints on correlations called Bell inequalities.

Surprisingly, nature violates Bell inequalities~\cite{AsDaRo82,GiVeWeHaHoPhSt15,ShMeChBiWaStGe15,HeBeDrReKaBlRu15,As15,GaFrKa14,RaHaHoGaFrLeLi18,AbAcAlAlAnAnBe18} which means that if the standard causal (or \textit{realist}) picture is to be maintained at least one of the remaining two assumptions, that is \textit{locality} or \textit{free} \textit{choice}, has to fail. It turns out that rejecting just one of those two assumptions is always enough to explain the observed correlations, while maintaining consistency with the causal structure imposed by the other. Either option poses a challenge to deep-rooted intuitions about reality, with a full range of viable positions open to serious philosophical dispute~\cite{Ma19,La19,No17}. Notably, quantum theory in its operational formulation does not provide any  clue regarding the causal structure at work, leaving such questions to the domain of interpretation. It is therefore interesting to ask about the extent to which a given assumption needs to be relaxed, if we insist on upholding the other one (while always maintaining \textit{realism}). In this paper, we seek to compare the cost of \textit{locality} and \textit{free choice} on an equal footing, without any preconceived conceptual biases. As a basis for comparison we choose to measure the weight of a given assumption in terms of the following question:

\textit{\parbox{0.92\columnwidth}{How often a given assumption, i.e. locality or free choice, can be retained, while  safeguarding the other assumption, in order to fully reproduce some given experimental statistics within a standard causal (or realist) approach?}}

\noindent This question presumes that a Bell experiment is performed trial-by-trial and the observed statistics can be explained in the standard causal model (or hidden variable) framework~\cite{Be93,Me93,Wi14a,BrCaPiScWe14,Sc19,Be87a,Pe09,WoSp15,ChKuBrGr15,Ca18}, which subsumes \textit{realism}. It means that the remaining two assumptions of \textit{locality} and \textit{free choice} translate into conditional independence between certain variables in the model, whose causal structure is determined by their spatio-temporal relations~\cite{Be87a,CoRe13a} (for some alternative approach endorsing indefinite causal structures see e.g.~\cite{Br14,AlBaHoLeSp17}, or~\cite{WhAr20} for discussion of retrocausality). Modelling of the experiment implies that in each run of the experiment all variables (including unobserved or hidden ones) always take definite values and the statistics accumulates over many trials. This leaves open the possibility that the violation of the assumptions do not have to occur on each run of the experiment to explain the given statistics. We can thus put flesh on the bones of the above question and seek the maximal proportion of trials in which a given assumption can be retained, while safeguarding the other assumption, so as to fully reproduce some given statistics. In the following, we shall denote so defined \textit{measure of locality} (safeguarding freedom of choice) as $\mu_{\scriptscriptstyle L}$ and \textit{measure of free choice} (safeguarding locality) as $\mu_{\scriptscriptstyle F}$. Also, without stating this in every instance, we note that in all subsequent discussion \textit{realism} is assumed.\footnote{As noted, realism is subsumed in the standard notion of causality, which is implicit in the definition of locality and free choice~\cite{Be93,Me93,Wi14a,BrCaPiScWe14,Sc19,Be87a}. So, henceforth, referring to the standard causal framework implies the realist approach. We also remark that, although, "realism" goes under different guises in the literature (e.g. "counterfactual definiteness", '"local causality", "hidden causes", etc.), for our purposes those distinctions are irrelevant and the underlying mathematics remains the same, i.e. it boils down to the hidden variable framework (which beyond physics is frequently referred to as the structural causal models~\cite{Pe09}). See~\cite{Wi14a,No07} for some discussion.}

There has been some previous research on this theme. A measure of locality analogous to $\mu_{\scriptscriptstyle L}$ was first proposed by Elitzur, Popescu and Rohrlich~\cite{ElPoRo92} to quantify non-locality in a singlet state. Note, it seems that the original idea of a bound for such a locality measure was expressed earlier, in~\cite{Ha91}, but a bound was not worked out. In any case, Elitzur, Popescu, and Rohrlich's measure was dubbed \textit{local} \textit{fraction} (or \textit{content}) and shown (with improvements in~\cite{BaKePi06,CoRe08,CoRe16}) to vanish in the limit of an infinite number of measurement settings. A substantial step was made in~\cite{PoBrGi12} where the local fraction is explicitly calculated for any pure two-qubit state for an arbitrary choice of settings. We note that those results concern measure $\mu_{\scriptscriptstyle L}$ only for the specific case of quantum-mechanical predictions. In this paper we go beyond this framework and consider the case of general experimental statistics (see~\cite{AbBaMa17} for extension to the idea of contextuality). To avoid confusion, the term \textit{local fraction} for measure $\mu_{\scriptscriptstyle L}$ will be only used in relation to the quantum case. Furthermore, we propose a similar treatment of the free choice assumption quantified by measure $\mu_{\scriptscriptstyle F}$. Natural as it may seem, this approach has not been pursued in the literature, with some other measures proposed to this effect~\cite{Br88,Ha10,Ha16,Ha10,Ha11b,BaGi11,PuRoBaLiGi14,AkTaMaPuThGi15,PuGi16,HaBr20} (all retaining locality as a principle, but departing from the original notion of free choice introduced by Bell~\cite{Be87a,CoRe13a}). 

We aim to comprehensively consider the extent to which a given assumption, i.e. \textit{locality} or \textit{free choice}, can be preserved through partial violation of the other assumption. To accomplish this,  we provide similar definitions and discuss on an equal footing both measures of locality $\mu_{\scriptscriptstyle L}$ and free choice $\mu_{\scriptscriptstyle F}$. Then, we derive the following results. First, we prove a general structural theorem about causal models explaining any given experimental statistics in a Bell experiment (for any number of settings) showing that such defined measures are necessarily equal, $\mu_{\scriptscriptstyle L}=\mu_{\scriptscriptstyle F}$. This result consolidates those two disparate concepts demonstrating their deep interchangeability. Second, we explicitly compute both measures for any \textit{non-signalling} statistics in a two-setting and two-outcome Bell scenario. This enables a direct interpretation to the amount of violation of the Clauser-Horne-Shimony-Holt (CHSH) inequalities~\cite{ClHoShHo69}. Third, we consider the special case of the quantum statistics with infinite number of settings, utilising  existing results for the local fraction $\mu_{\scriptscriptstyle L}$, which thus translate on the newly developed concept of the measure of free choice $\mu_{\scriptscriptstyle F}$. Fig.~\ref{Fig_Summary} summarises the results in the paper.

\begin{figure}[t]

\centering
\includegraphics[width=1\columnwidth]{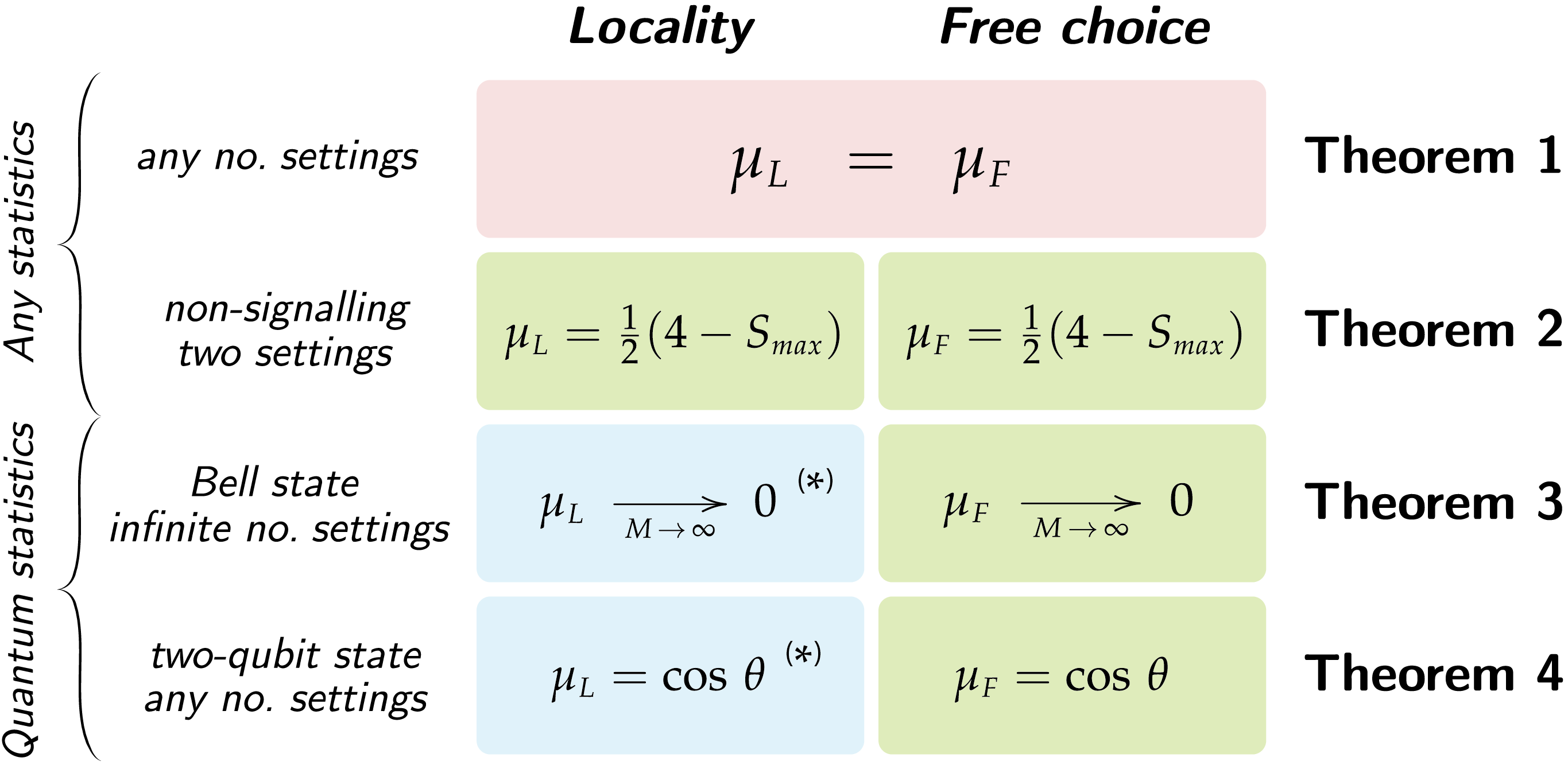}
\caption{\label{Fig_Summary}{\bf\textsf{\mbox{Summary of the results.}}} The main  \textbf{Theorem~\ref{Theorem-equivalence}} is the backbone of the paper, consolidating both measures of locality $\mu_{\scriptscriptstyle L}$ and free choice $\mu_{\scriptscriptstyle F}$.  \textbf{Theorem~\ref{Theorem-CHSH}} is a theory-independent result about both measures $\mu_{\scriptscriptstyle L}$ and $\mu_{\scriptscriptstyle F}$. It offers a concrete interpretation for the amount of violation of the CHSH inequalities. \textbf{Theorems~\ref{Theorem-ChainedBell}} and \textbf{\ref{Theorem-Portmann}} are specific to the quantum-mechanical statistics stated here for measure $\mu_{\scriptscriptstyle F}$. They are translations of some remarkable local fraction results $\mu_{\scriptscriptstyle L}$ in the literature (marked with an $^{(*)}$, cf.~\cite{BaKePi06,CoRe08,CoRe16,PoBrGi12}).}
\end{figure}

\section*{Results}

\subsection*{Bell experiment and Fine's theorem}

Let us consider the usual Bell-type scenario with two parties, called Alice and Bob, playing the roles of agents conducting experiments on two separated systems (whose nature is irrelevant for the argument). We assume that on each side there are two possible outcomes labelled respectively $a,b=\pm1$ and $M$ possible measurement settings labelled respectively $x,y\in\mathfrak{M}$ where $\mathfrak{M}\equiv\{1,2\,,\dots,M\}$. A Bell experiment consists of a series of trials in which Alice and Bob each choose a setting and make a measurement registering the outcome. After many repetitions, they compare their results described by the set of $M\times M$ distributions $\{P_{\scriptscriptstyle  ab|xy}\}_{\scriptscriptstyle xy}$\,, where $P_{\scriptscriptstyle  ab|xy}$ denotes the probability of obtaining outcomes $a,b$, given measurements $x,y$ were made on Alice and Bob's side respectively. For conciseness, following the terminology in~\cite{Sc19}, we will call  $\{P_{\scriptscriptstyle  ab|xy}\}_{\scriptscriptstyle xy}$ a \textit{"behaviour"}. Note that without assuming anything about the causal structure underlying the experiment any behaviour is admissible (as long as the distributions are normalised, i.e. $\sum_{\scriptscriptstyle a,b}P_{\scriptscriptstyle  ab|xy}=1$ for each $x,y\in\mathfrak{M}$). In particular, quantum theory gives a  prescription for calculating the experimental statistics $P_{\scriptscriptstyle  ab|xy}$ for each choice of settings $x,y\in\mathfrak{M}$ based on the formalism of Hilbert spaces.

It is instructive to recall the special case of two measurement settings on each side $x,y\in\mathfrak{M}=\{0,1\}$ for which Bell derived his seminal result. Briefly, this can be expressed by saying that any \textit{local} hidden variable model with \textit{free choice} has to satisfy the following four CHSH inequalities~\cite{ClHoShHo69}
\begin{eqnarray}\label{Bell-CHSH-inequalities}
|S_{\scriptscriptstyle  i}|\ \leqslant\ 2&&\qquad\text{for\ \  $i=1,...\,,4$}\,,
\end{eqnarray}
where
\begin{eqnarray}\label{S1}
S_{\scriptscriptstyle  1}&\!\!=\!\!&\ \ \ \langle ab\rangle_{\scriptscriptstyle  00}+\langle ab\rangle_{\scriptscriptstyle  01}+\langle ab\rangle_{\scriptscriptstyle  10}-\langle ab\rangle_{\scriptscriptstyle  11}\,,\\\label{S2}
S_{\scriptscriptstyle  2}&\!\!=\!\!&\ \ \ \langle ab\rangle_{\scriptscriptstyle  00}+\langle ab\rangle_{\scriptscriptstyle  01}-\langle ab\rangle_{\scriptscriptstyle  10}+\langle ab\rangle_{\scriptscriptstyle  11}\,,\\\label{S3}
S_{\scriptscriptstyle  3}&\!\!=\!\!&\ \ \ \langle ab\rangle_{\scriptscriptstyle  00}-\langle ab\rangle_{\scriptscriptstyle  01}+\langle ab\rangle_{\scriptscriptstyle  10}+\langle ab\rangle_{\scriptscriptstyle  11}\,,\\\label{S4}
S_{\scriptscriptstyle  4}&\!\!=\!\!&-\langle ab\rangle_{\scriptscriptstyle  00}+\langle ab\rangle_{\scriptscriptstyle  01}+\langle ab\rangle_{\scriptscriptstyle  10}+\langle ab\rangle_{\scriptscriptstyle  11}\,,
\end{eqnarray}
with $\langle ab\rangle_{\scriptscriptstyle  xy}=\sum_{\scriptscriptstyle a,b}\,ab\,P_{\scriptscriptstyle  ab|xy}$ being correlation coefficients for a given choice of settings $x,y$. Interestingly, by virtue of Fine's theorem~\cite{Fi82a,Ha14b}, this is also a sufficient condition for a \textit{non-signalling} behaviour $\{P_{\scriptscriptstyle  ab|xy}\}_{\scriptscriptstyle  xy}$ to be explained by a \textit{local} hidden variable model with \textit{freedom of choice} (for non-signalling see Eqs.~(\ref{non-signalling-Alice}) and (\ref{non-signalling-Bob})).

It is crucial to observe that, although locality and freedom of choice are two disparate concepts with different ramifications for our understanding of the experiment, they are in a certain sense interchangeable. If locality is dropped with Alice and Bob freely choosing their settings, then the boxes, by influencing one another, can produce any behaviour $\{P_{\scriptscriptstyle  ab|xy}\}_{\scriptscriptstyle  xy}$. Similarly, a violation of the free choice assumption can be used to reproduce any behaviour $\{P_{\scriptscriptstyle  ab|xy}\}_{\scriptscriptstyle  xy}$, without giving up locality. It is straightforward to see how this might work if one of the two assumptions fails on \textit{every} experimental trial.\footnote{For the simulation of a given behaviour $\{P_{\scriptscriptstyle  ab|xy}\}_{\scriptscriptstyle  xy}$ in a Bell experiment one may proceed as follows. Upon rejection of locality, in \textit{each}\ trial the system on Alice's side, one may not only use input $x$ but also $y$ to generate outcomes (and similarly for the box on Bob's side) that comply with the appropriate distribution. On the other hand, when freedom of choice is abandoned, both settings $x,y$ may be specified in advance on \textit{each} trial and the boxes can be instructed to provide the outcomes needed to simulate the appropriate distribution. It is however unclear how this might work with \textit{occasional} violation of the respective assumptions.}



However, such a complete renouncement of assumptions so central to our view of nature may seem excessive, especially when the CHSH inequalities are violated only by a little amount (less than the maximal algebraic bound of $|S_{\scriptscriptstyle  i}|\leqslant4$), leaving room for a possible explanation of the experimental statistics by rejecting one of the assumptions \textit{sometimes} only. Here we assess the cost of such a partial violation by asking how often a given assumption can be retained in order to account for a  behaviour $\{P_{\scriptscriptstyle  ab|xy}\}_{\scriptscriptstyle  xy}$\,. We will investigate both cases in parallel: (\ref{locality-informal}) full freedom of choice with \textit{occasional} non-locality (communication), and (\ref{freedom-informal}) the possibility of retaining full locality at a price of compromising freedom of choice (by controlling or rigging measurement settings) on \textit{some} of the trials. 
We shall use the least frequency of violation, required to model some statistics with a hypothetical simulation, as a natural figure of merit, guided by the principle that the less the violation the better. Notably, such simulations should not restrict possible distributions of measurement settings $P_{\scriptscriptstyle  xy}$. In other words, we define a \textit{measure of locality} $\mu_{\scriptscriptstyle  L}$ as
\begin{equation}\label{locality-informal}
    \textit{\parbox{.85\columnwidth}{the \underline{maximal} fraction of trials in which Alice and Bob do not need to communicate trying to simulate a given behaviour $\{P_{\scriptscriptstyle  ab|xy}\}_{\scriptscriptstyle  xy}$\,, optimised over \underline{all} conceivable strategies with freely chosen settings.}}\quad \tag{$\scriptstyle \spadesuit$}
\end{equation}
Similarly, we define a \textit{measure of free choice} $\mu_{\scriptscriptstyle  F}$ as
\begin{equation}\label{freedom-informal}
    \textit{\parbox{.85\columnwidth}{the \underline{maximal} fraction of trials in which Alice and Bob can grant free choice of settings in trying to simulate a given behaviour $\{P_{\scriptscriptstyle  ab|xy}\}_{\scriptscriptstyle  xy}$\,, optimised over \underline{all} conceivable local strategies.}}\quad\tag{$\scriptstyle\clubsuit$}
\end{equation}
In the quantum-mechanical context the measure $\mu_{\scriptscriptstyle  L}$ is  called a \textit{local fraction}~\cite{ElPoRo92,Ha91,BaKePi06,CoRe08,CoRe16,PoBrGi12}. By analogy, when considering the quantum-mechanical statistics the measure $\mu_{\scriptscriptstyle  F}$ might be called a \textit{free fraction}. This provides an equal basis for comparing the two assumptions within the standard causal (or realist) approach, which we formalise in the following section.

\subsection*{Causal models, locality and free choice}

The appropriate framework for the discussion of locality and free choice is provided by hidden variable models~\cite{Be93,Me93,Wi14a,BrCaPiScWe14,Sc19}. First, a hidden variables model allows a formal statement of the \textit{realism} assumption, understood to mean that properties of a physical system exist irrespective of an act of measurement (counterfactual definiteness). Second, hidden variable models provide the causal language in which the locality and free choice assumptions are expressed~\cite{Be87a,Pe09}. The \textit{locality} assumption conveys the requirement that the propagation of physical (or causal) influences have to follow the spatio-temporal structure of events (i.e., preserve the arrow of time and respect that actions at a distance require time). The \textit{free} \textit{choice} assumption concerns the choice of measurement settings which are deemed cusally unaffected by anything in the past (and thus it is sometimes called \textit{measurement} \textit{independence}).\footnote{As noted, the \textit{free} \textit{choice} assumption is sometimes called \textit{measurement} \textit{independence}. Instead of on the agent, measurement independence is focussed on the measurement devices and possible correlations between their settings, which can affect the observed statistics. Regardless of interpretation, the mathematics remains the same, with the source of correlations traced to some common factor (in the causal past).} Both assumptions take the form of conditional independencies between certain variables in a hidden variables model.

To make this idea more concrete, let us consider a given set of probability distributions (behaviour) $\{P_{\scriptscriptstyle ab|xy}\}_{\scriptscriptstyle xy}$ which describes the statistics in a Bell experiment. Without loss of generality, by conditioning on $\lambda$ in some \textit{a priori} unknown hidden variable space $\Lambda$, one can always write~\cite{BrCaPiScWe14,Sc19,Pe09}
\begin{eqnarray}\label{Pab|xy}
P_{\scriptscriptstyle ab|xy}\ =\ \sum_{\scriptscriptstyle \lambda\in\Lambda}P_{\scriptscriptstyle ab|xy\lambda}\cdot P_{\scriptscriptstyle \lambda|xy}\,,
\end{eqnarray}
where $P_{\scriptscriptstyle \lambda|xy}$ and $P_{\scriptscriptstyle ab|xy\lambda}$ are valid (i.e. normalised) conditional probability distributions. The role of the hidden variable (cause in the past) $\lambda\in\Lambda$, distributed according to some $P_{\scriptscriptstyle \lambda}$, is to provide an explanation of the observed experimental statistics. This means that at each run of the experiment the outcomes are described by the distribution $P_{\scriptscriptstyle ab|xy\lambda}$ with $\lambda\in\Lambda$ fixed in a given trial, so that the accumulated experimental statistics $P_{\scriptscriptstyle ab|xy}$ obtains by sampling from some distribution $P_{\scriptscriptstyle \lambda|xy}$ over the whole hidden variable space $\Lambda$. It is customary to say that 
\begin{equation}\label{HV}
	\textit{\parbox{.81\columnwidth}{the choice of space $\Lambda$ and probability distribution $P_{\scriptscriptstyle \lambda}$ along with conditional distributions $P_{\scriptscriptstyle ab |xy\lambda}$ and $P_{\scriptscriptstyle \lambda|xy}$ satisfying Eq.~(\ref{Pab|xy}) specify a \underline{hidden variable} (HV) model of a given behaviour  $\{P_{\scriptscriptstyle ab|xy}\}_{\scriptscriptstyle xy}$\,.}}\quad\tag{$\star$}
\end{equation}
Note that such a model implicitly describes the distribution of settings chosen by Alice and Bob through the standard formula
\begin{eqnarray}\label{distribution-settings}
P_{\scriptscriptstyle xy}\ =\ \sum_{\scriptscriptstyle \lambda\in\Lambda}P_{\scriptscriptstyle xy|\lambda}\cdot P_{\scriptscriptstyle \lambda}\,.
\end{eqnarray}

\begin{figure}[t]
\centering
\includegraphics[width=0.95\columnwidth]{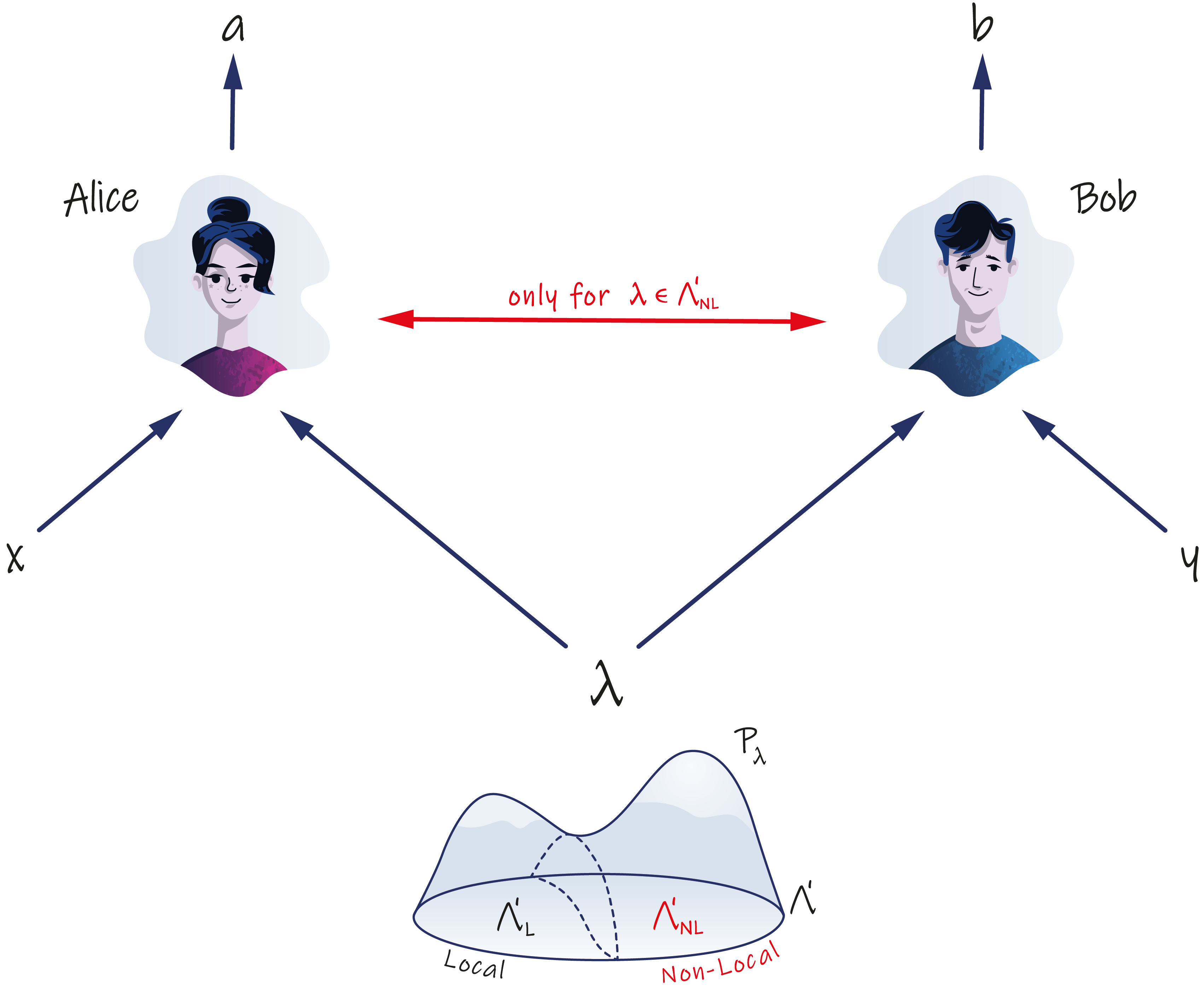}
\caption
{\label{Fig_CausalDiag_Locality}{\bf\textsf{Causal model with some non-locality (communication).}} In a Bell scenario, with free choice of settings, correlations between Alice and Bob's outcomes have two possible explanations: common cause in the past or causal influence between the parties. In any causal model the space of hidden variables (representing common causes) splits into two disjoint parts $\Lambda'=\Lambda'_{\scriptscriptstyle L}\cup\Lambda'_{\scriptscriptstyle NL}$ distinguished by whether, for a given $\lambda\in\Lambda'$, causal influence occurs or not, Eq.~(\ref{locality-splitting}). Then, \textit{locality} is measured by the proportion of events when locality is maintained, which is equal to the probability accumulated over subset $\Lambda'_{\scriptscriptstyle L}$, i.e. ${{Prob}}\,(\lambda\in\Lambda'_{\scriptscriptstyle L})\,\equiv\,\sum_{\scriptscriptstyle \lambda\in\Lambda'_{\scriptscriptstyle L}}P_{\scriptscriptstyle \lambda}$.}
\end{figure}

So far the framework is general enough to accommodate \textit{any} causal explanation of the statistics observed in the experiment. The assumptions of locality and free choice take the form of constraints on conditional distributions in~(\ref{HV}). For a \textit{local hidden variable} (LHV) model, we require the following factorisation\footnote{Locality can be seen as a conjunction of two conditions: \textit{parameter independence} $P_{\scriptscriptstyle a|xy\lambda}=P_{\scriptscriptstyle a|x\lambda}$ \& $P_{\scriptscriptstyle b|xy\lambda}=P_{\scriptscriptstyle b|y\lambda}$, and \textit{outcome independence} $P_{\scriptscriptstyle a|bxy\lambda}=P_{\scriptscriptstyle a|xy\lambda}$ \& $P_{\scriptscriptstyle b|axy\lambda}=P_{\scriptscriptstyle b|xy\lambda}$. One can show that such defined locality entails the factorisation condition $P_{\scriptscriptstyle ab|xy\lambda}=P_{\scriptscriptstyle a|x\lambda}\cdot P_{\scriptscriptstyle b|y\lambda}$~\cite{Ja84}.}
\begin{eqnarray}\label{factorisation}
P_{\scriptscriptstyle ab|xy\lambda}\ =\ P_{\scriptscriptstyle a|x\lambda}\cdot P_{\scriptscriptstyle b|y\lambda}\,,
\end{eqnarray}
for each $x,y\in\mathfrak{M}$ and all $\lambda\in\Lambda$. The \textit{freedom of choice} assumption consists of requiring that $\lambda$ does not contain any information about variables $x,y$ representing Alice and Bob's choice of measurement settings. This boils down to the independence condition~\cite{Be87a,CoRe13a}
\begin{eqnarray}\label{free-choice}
\qquad\qquad P_{\scriptscriptstyle \lambda|xy}\ =\ P_{\scriptscriptstyle \lambda}\qquad\text{(or equivalently\ \  $P_{\scriptscriptstyle xy|\lambda}\,=\,P_{\scriptscriptstyle xy}$)\,,}
\end{eqnarray}
holding for $x,y\in\mathfrak{M}$ and  all $\lambda\in\Lambda$. In the following, we will abbreviate a \textit{hidden variable} model with \textit{freedom of choice} as FHV model.

\begin{figure}[t]
\centering
\includegraphics[width=0.95\columnwidth]{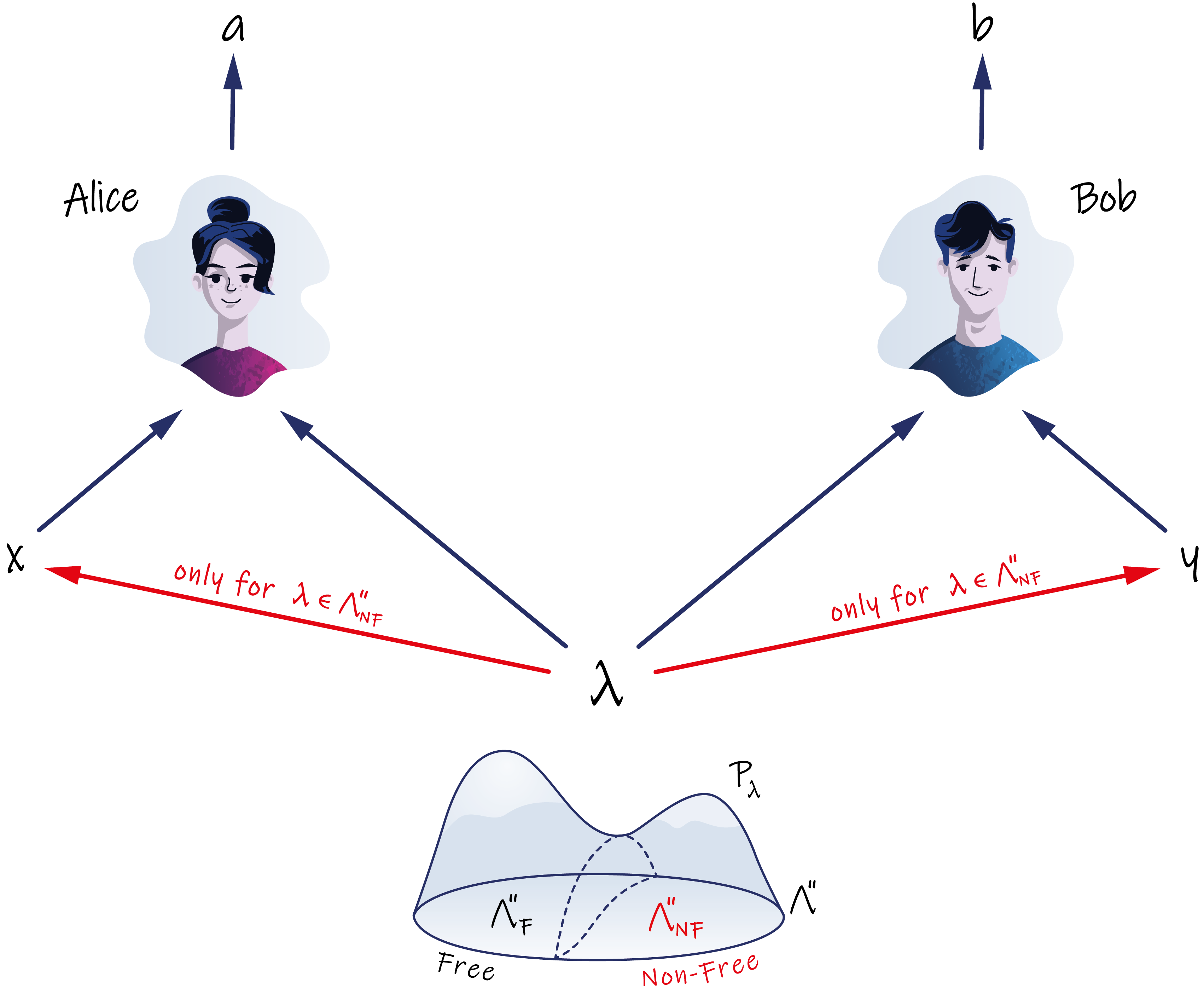}
\caption
{\label{Fig_CausalDiag_FreeChoice}{\bf\textsf{Causal model with some freedom of choice (rigging).}} In a Bell scenario, with locality assumption, correlations between the outcomes on Alice and Bob's side can be explained by a common cause affecting choice or not (the latter implies freedom of choice). In any causal model the space of hidden variables (representing common causes) splits into two disjoint parts $\Lambda''=\Lambda''_{\scriptscriptstyle F}\cup\Lambda''_{\scriptscriptstyle NF}$ distinguished by whether, for a given $\lambda\in\Lambda''$, the choice is free or not, Eqs.~(\ref{freedom-splitting}). Then, the parties enjoy \textit{freedom of choice} only on the trials when $\lambda\in\Lambda''_{\scriptscriptstyle F}$, which happens with a frequency equal to the probability accumulated over subset $\Lambda''_{\scriptscriptstyle F}$, i.e. ${{Prob}}\,(\lambda\in\Lambda''_{\scriptscriptstyle F})\,\equiv\,\sum_{\scriptscriptstyle \lambda\in\Lambda''_{F}}P_{\scriptscriptstyle \lambda}$.}
\end{figure}

The crucial point  is the distinction between \textit{local} vs \textit{non-local} as well as \textit{free} vs \textit{non-free} situations in the individual runs of the experiment modelled by Eq.~(\ref{Pab|xy}).
This means that each condition Eq.~(\ref{factorisation}) and Eq.~(\ref{free-choice}) should be considered separately for each $\lambda\in\Lambda$, i.e. whenever the respective condition does not hold for a given $\lambda$ the assumption fails on the corresponding experimental trials. Such a distinction leads to a natural splitting of the underlying HV space into two \textit{unique} partitions $\Lambda=\Lambda_{\scriptscriptstyle L}\cup\Lambda_{\scriptscriptstyle NL}$ and $\Lambda=\Lambda_{\scriptscriptstyle F}\cup\Lambda_{\scriptscriptstyle NF}$. The first one divides $\Lambda$ by the locality property
\begin{eqnarray}\label{locality-splitting}
\begin{array}{lcl}
\lambda\in\Lambda_{\scriptscriptstyle L}&\ \Leftrightarrow\ &\text{Eq.~(\ref{factorisation})}\ \ \text{holds for \textit{all} $x,y\in\mathfrak{M}$}\,,\vspace{0.1cm}\\
\lambda\in\Lambda_{\scriptscriptstyle NL}&\ \Leftrightarrow\ &\text{Eq.~(\ref{factorisation})}\ \ \text{fails for \textit{some} $x,y\in\mathfrak{M}$}\,,
\end{array}
\end{eqnarray}
while the second one divides $\Lambda$ by the free choice property
\begin{eqnarray}\label{freedom-splitting}
\begin{array}{lcl}
\lambda\in\Lambda_{\scriptscriptstyle F}&\ \Leftrightarrow\ &\text{Eq.~(\ref{free-choice})}\ \ \text{holds for \textit{all} $x,y\in\mathfrak{M}$}\,,\vspace{0.1cm}\\
\lambda\in\Lambda_{\scriptscriptstyle NF}&\ \Leftrightarrow\ &\text{Eq.~(\ref{free-choice})}\ \ \text{fails for \textit{some} $x,y\in\mathfrak{M}$}\,.
\end{array}
\end{eqnarray}

Figs.~\ref{Fig_CausalDiag_Locality} and \ref{Fig_CausalDiag_FreeChoice} illustrate the causal structures for two extreme cases: FHV and LHV models (in general built on different HV spaces $\Lambda'$ and $\Lambda''$). The first one grants full freedom of choice ($\Lambda'=\Lambda'_{\scriptscriptstyle F}$) while allowing for partial violation of locality ($\Lambda'\supset\Lambda'_{\scriptscriptstyle L}$). The second one retains full locality ($\Lambda''=\Lambda''_{\scriptscriptstyle L}$) while admitting some violation of free choice ($\Lambda''\supset\Lambda''_{\scriptscriptstyle F}$).

Thus, for a given experimental trial (with $\lambda\in\Lambda$ fixed) the constraints in Eqs. (\ref{locality-splitting}) and (\ref{freedom-splitting}) indicate, respectively, whether some non-local influence between the parties takes place ($\lambda\in\Lambda_{\scriptscriptstyle NL}$) and whether some influence from the past on the measurement settings occurs ($\lambda\in\Lambda_{\scriptscriptstyle NF}$). In other words, in a hypothetical simulation scenario these possibilities correspond to, respectively, communication or rigging measurement settings. How often this has to happen depends on the distribution $P_{\scriptscriptstyle \lambda}$. This picture lends itself to quantifying the degree of locality and freedom choice in a given HV model.

\begin{remark}\label{remark-definition}
For a given HV model (\ref{HV}) locality is measured by
$\textsl{{Prob}}\,(\lambda\in\Lambda_{\scriptscriptstyle L})\,\equiv\,\sum_{\scriptscriptstyle \lambda\in\Lambda_{L}}P_{\scriptscriptstyle \lambda}$, and similarly freedom of choice is measured by
$\textsl{{Prob}}\,(\lambda\in\Lambda_{\scriptscriptstyle F})\,\equiv\,\sum_{\scriptscriptstyle \lambda\in\Lambda_{F}}P_{\scriptscriptstyle \lambda}$.
\end{remark}
\noindent This remark captures the intuition of measuring locality and freedom of choice by considering the proportion of trials when the respective property is maintained  across the whole experimental ensemble.  We note that this quantity is model-dependent, since it is a property of a particular HV model adopted to explain some given experimental statistics $\{P_{\scriptscriptstyle ab|xy}\}_{\scriptscriptstyle xy}$ (including the distribution of measurement settings $P_{\scriptscriptstyle xy}$, cf. Eq.~(\ref{distribution-settings})).

The concepts just introduced allow a precise expression for the informal definitions~(\ref{locality-informal}) and (\ref{freedom-informal}) given above.

\begin{definition}\label{definition}
For a given behaviour $\{P_{\scriptscriptstyle ab|xy}\}_{\scriptscriptstyle xy}$ the measure of locality $\mu_{\scriptscriptstyle L}$ and freedom of choice $\mu_{\scriptscriptstyle F}$ are defined as 
\begin{eqnarray}\label{locality-measure}
\mu_{\scriptscriptstyle L}\ :=\ \min_{\scriptscriptstyle P_{\scriptscriptstyle xy}}\,\max_{\scriptscriptstyle FHV}\sum_{\scriptscriptstyle \lambda\in\Lambda_{L}}P_{\scriptscriptstyle \lambda}\,\,,
\\\label{freedom-measure}
\mu_{\scriptscriptstyle F}\ :=\ \min_{\scriptscriptstyle P_{\scriptscriptstyle xy}}\,\max_{\scriptscriptstyle LHV}\sum_{\scriptscriptstyle \lambda\in\Lambda_{F}}P_{\scriptscriptstyle \lambda}\,\,,
\end{eqnarray}
where the maxima are taken respectively over \underline{all}  hidden variable models with freedom of choice (FHV) or \underline{all} local hidden variable models (LHV) simulating given behaviour $\{P_{\scriptscriptstyle ab|xy}\}_{\scriptscriptstyle xy}$, with a fixed distribution of settings $P_{\scriptscriptstyle xy}$\,, minimized over \underline{any} choice of the latter.
\end{definition}

This definition follows the intuition of, respectively, locality or free choice as properties that can be relaxed only to the extent that is required to maintain the other assumption in every experimental situation (i.e., for any distribution of measurement settings $P_{\scriptscriptstyle xy}$). Formally, the measures $\mu_{\scriptscriptstyle L}$ and $\mu_{\scriptscriptstyle F}$ count the maximal frequency of, respectively, local or free choice events optimised over all protocols simulating $\{P_{\scriptscriptstyle ab|xy}\}_{\scriptscriptstyle xy}$ without violating of the other assumption, cf. \textbf{Remark~\ref{remark-definition}}. The minimum over all $P_{\scriptscriptstyle xy}$ amounts to the worst case scenario, which takes into account the possibility that $P_{\scriptscriptstyle xy}$ is \textit{a priori} unspecified (i.e., this amount of freedom is enough to simulate an experiment with any arbitrary choice of distribution $P_{\scriptscriptstyle xy}$ in compliance with Eq.~(\ref{distribution-settings})).

At first glance, even if conceptually appropriate, such a definition might seem too general to provide a manageable notion, due to the range of experimental scenarios that need to be taken into account (i.e. arbitrariness of $P_{\scriptscriptstyle xy}$). However, the situation considerably simplifies because of the following lemma (see \textbf{Methods} section for further discussion and proof). This lemma also provides additional support for \textbf{Definition~\ref{definition}}.

\begin{lemma}\label{max-lemma}
In both Eqs.~(\ref{locality-measure}) and (\ref{freedom-measure}) in \textbf{Definition~\ref{definition}} the first minimum can be omitted, i.e. we have
\begin{eqnarray}\label{locality-measure-max}
\mu_{\scriptscriptstyle L}\ =\ \max_{\scriptscriptstyle FHV}\sum_{\scriptscriptstyle \lambda\in\Lambda_{L}}P_{\scriptscriptstyle \lambda}\,\,,\\\label{freedom-measure-max}
\mu_{\scriptscriptstyle F}\ =\ \max_{\scriptscriptstyle LHV}\sum_{\scriptscriptstyle \lambda\in\Lambda_{F}}P_{\scriptscriptstyle \lambda}\,\,,
\end{eqnarray}
where the respective maxima are taken for some fixed nontrivial distribution $P_{\scriptscriptstyle xy}$ (i.e., the expression is insensitive to this choice provided all settings are probed, $P_{\scriptscriptstyle xy}\neq0$ for all $x,y$).
\end{lemma}

It is in this way that the present measure of locality $\mu_{\scriptscriptstyle L}$ extends the notion of \textit{local} \textit{fraction}~\cite{ElPoRo92,Ha91,BaKePi06,CoRe08,CoRe16,PoBrGi12} to arbitrary experimental behaviour $\{P_{\scriptscriptstyle ab|xy}\}_{\scriptscriptstyle xy}$. Remarkably, the twin concept, which is the measure of free choice $\mu_{\scriptscriptstyle F}$ has not been considered at all. Perhaps the reason for this omission is the issue of arbitrariness of the distribution $P_{\scriptscriptstyle xy}$, for which there are non-trivial constraints when freedom of choice is violated (note that for the measure $\mu_{\scriptscriptstyle L}$ this problem does not occur). Those concerns can be dismissed only after the proper treatment in \textbf{Lemma~\ref{max-lemma}}. This allows a so defined measure of freedom $\mu_{\scriptscriptstyle F}$ on a par with the more familiar  measure of locality $\mu_{\scriptscriptstyle L}$. 

So far the concepts of violation of locality and freedom of choice, and the corresponding measures $\mu_{\scriptscriptstyle L}$ and $\mu_{\scriptscriptstyle F}$, have been kept separate. This is expected given  their disparate character. First, each concept plays a different role in the description of an experiment and hence offers a different explanation for any observed correlations, this is, direct influence (communication during the experiment) vs measurement dependence (employing common past for rigging measurement settings). Second, on the level of causal modelling those assumptions are expressed differently, Eq.~(\ref{factorisation}) vs Eq.~(\ref{free-choice}). Third, violating free choice gives rise to subtle issues regarding constraints on the distribution of settings $P_{\scriptscriptstyle xy}$ (as noted, these concerns are addressed in \textbf{Lemma~\ref{max-lemma}}).

Having brought all those issues to the spotlight, it is surprising that the assumption of locality and free choice are intrinsically connected. We now present the key result in this paper showing the exchangeability of both concepts, while maintaining the same degree of locality and freedom of choice  so defined. It holds for any number of settings $x,y\in\mathfrak{M}=\{0,1,...\,,M\}$ (see \textbf{Methods} for the proof).

\begin{theorem}\label{Theorem-equivalence}
For a given behaviour $\{P_{\scriptscriptstyle ab|xy}\}_{\scriptscriptstyle xy}$ the degree of locality and freedom of choice are the same, i.e. both measures in \textbf{Definition~\ref{definition}} coincide $\mu_{\scriptscriptstyle L}=\mu_{\scriptscriptstyle F}$.
\end{theorem}

This is a general structural theorem about causal modelling of a given behaviour $\{P_{\scriptscriptstyle ab|xy}\}_{\scriptscriptstyle xy}$. It means that the resources measured by the frequency of causal interventions of either sort, required to explain an experimental statistics, are equally costly. Thus, as far as the statistics is concerned,  causal explanations resorting either to violation of locality or free choice (or measurement dependence) should be kept on an equal footing. Preference should be guided by a better understanding of a particular situation (design of the experiment as well as ontological commitments in its description).

Let us emphasise two features of \textbf{Theorem~\ref{Theorem-equivalence}}. First, this is a \textit{theory-independent} result in the sense that it applies directly to experimental statistics irrespective of the design or theoretical framework behind the experiment (with the quantum predictions being just one example). Second, the connection between those two seemingly disparate quantities $\mu_{\scriptscriptstyle L}$ and $\mu_{\scriptscriptstyle F}$ has a practical advantage: knowledge of one suffices to compute the other. Both features are illustrated by the following results.

\subsection*{Non-signalling behaviour with binary settings}

Consider the case of Bell's experiment with only two measurement settings on each side $x,y\in\mathfrak{M}=\{0,1\}$.
Let us recall that \textit{non-signalling} of some given behaviour $\{P_{\scriptscriptstyle  ab|xy}\}_{\scriptscriptstyle  xy}$ means that Alice \textit{cannot} infer Bob's measurement setting (whether it is $y=0$ or $1$) from the statistics on her side alone, i.e.
\begin{eqnarray}\label{non-signalling-Alice}
P_{\scriptscriptstyle a|x0}\ =\ \sum_{\scriptscriptstyle b}P_{\scriptscriptstyle ab|x0}\ =\ \sum_{\scriptscriptstyle b}P_{\scriptscriptstyle ab|x1}\ =\ P_{\scriptscriptstyle a|x1}\quad\text{ for all $a,x$}\,,
\end{eqnarray}
and similarly on Bob's side (whether Alice chooses $x=0$ or $1$), i.e.
\begin{eqnarray}\label{non-signalling-Bob}
P_{\scriptscriptstyle b|0y}\ =\ \sum_{\scriptscriptstyle a}P_{\scriptscriptstyle ab|0y}\ =\ \sum_{\scriptscriptstyle a}P_{\scriptscriptstyle ab|1y}\ =\ P_{\scriptscriptstyle b|1y}\quad\text{ for all $b,y$}\,.
\end{eqnarray}

Now we can state another result which explicitly computes both measures $\mu_{\scriptscriptstyle L}$ and $\mu_{\scriptscriptstyle F}$ in a surprisingly simple form (see \textbf{Methods} for the proof).
\begin{theorem}\label{Theorem-CHSH}
For a given non-signalling behaviour $\{P_{\scriptscriptstyle ab|xy}\}_{\scriptscriptstyle xy}$ with binary settings $x,y\in\mathfrak{M}=\{0,1\}$ both measures of locality $\mu_{\scriptscriptstyle L}$ and free choice $\mu_{\scriptscriptstyle F}$ from \textbf{Definition~\ref{definition}} are equal to 
\begin{eqnarray}\label{Theorem-CHSH-Equation}
\mu_{\scriptscriptstyle L}\ =\ \mu_{\scriptscriptstyle F}\ =\ \left\{\begin{array}{lll}\tfrac{1}{2}(4-S_{\scriptscriptstyle max})\,,\ &&\text{if $S_{\scriptscriptstyle max}>2$\,,}\vspace{0.15cm}\\
1\,,&&\text{otherwise\,,}\end{array}\right.
\end{eqnarray}
where $S_{\scriptscriptstyle max}=\max\,\{|S_{\scriptscriptstyle i}|:i=1,...\,,4\}$ is the maximum absolute value of the four CHSH expressions in Eqs.~(\ref{S1})-(\ref{S4}).
\end{theorem}

We thus obtain a systematic method for assessing the degree of locality and free choice directly from the observed statistics $\{P_{\scriptscriptstyle ab|xy}\}_{\scriptscriptstyle xy}$ without reference to the specifics of the experiment (the only requirement is non-signalling of the observed distributions). In this sense, this is a general \textit{theory-independent} statement. 

Overall, \textbf{Theorem~\ref{Theorem-CHSH}} allows an interpretation of the amount of violation of the CHSH inequalities in Bell-type experiments as a fraction of trials violating locality (granted freedom of choice) or equivalently trials without freedom of choice (given locality).

\subsection*{The quantum case: Binary settings and beyond}

Let us restrict our attention to the special case of the quantum statistics.
Notably, various aspects of non-locality have been extensively researched in relation to the quantum-mechanical predictions, see~\cite{BrCaPiScWe14,Sc19} for a review. This includes the notion of \textit{local fraction}~\cite{ElPoRo92,Ha91,BaKePi06,CoRe08,CoRe16,PoBrGi12}, which is the same as measure $\mu_{\scriptscriptstyle L}$ here defined for a general behaviour $\{P_{\scriptscriptstyle  ab|xy}\}_{\scriptscriptstyle  xy}$. As noted, it may be thus surprising that the equally natural measure of freedom $\mu_{\scriptscriptstyle F}$ has not been explored. \textbf{Theorem~\ref{Theorem-equivalence}} bridges the gap between those two seemingly disparate notions: there is no actual need for separate study. We next review some crucial results for the \textit{local fraction} in the quantum-mechanical framework, which allows us to make similar statements for the measure of freedom $\mu_{\scriptscriptstyle F}$.

We first observe that \textbf{Theorem~\ref{Theorem-CHSH}} can be readily applied to the quantum-mechanical statistics (where non-signalling holds).  In a Bell experiment, quantum probabilities obtain through the standard formula $P_{\scriptscriptstyle ab|xy}=Tr\,[\,\rho\,\mathbb{P}_{\scriptscriptstyle x}^{\scriptscriptstyle a}\otimes\mathbb{P}_{\scriptscriptstyle y}^{\scriptscriptstyle b}\,]$ where $\rho$ is a (bipartite) mixed state with two PVMs $\{\mathbb{P}_{\scriptscriptstyle x}^{\scriptscriptstyle a=\pm1}\}$ and $\{\mathbb{P}_{\scriptscriptstyle y}^{\scriptscriptstyle b=\pm1}\}$ representing Alice and Bob's choice of measurement settings $x,y\in\mathfrak{M}=\{0,1\}$. Calculating the CHSH expressions Eqs.~(\ref{S1})-(\ref{S4}) in each particular case is straightforward, which gives explicitly the expression for both measures $\mu_{\scriptscriptstyle L}$ and $\mu_{\scriptscriptstyle F}$ via Eq.~(\ref{Theorem-CHSH-Equation}). The result of special significance concerns the famous Tsirelson bound $S_{\scriptscriptstyle max}^{\scriptscriptstyle  QM}=2\sqrt{2}$ for the maximal violation of the CHSH inequalities in quantum mechanics~\cite{Ts80}. By virtue of  \textbf{Theorem~\ref{Theorem-CHSH}}, this means that  in order to locally recover the quantum predictions in a Bell experiment with two settings, Alice and Bob can enjoy freedom of choice in the worst case, at most, with a fraction $\mu_{\scriptscriptstyle F}=2-\sqrt{2}\approx0.59$ of all trials (corresponding to the choice of measurements on a maximally entangled state that saturate the Tsirelson bound). Clearly, the same applies to local fraction $\mu_{\scriptscriptstyle L}$ in a two-setting scenario.

Interestingly, relaxing the constraint on the number of settings for Alice and Bob's measurements $x,y\in\mathfrak{M}=\{1,\,2,\,3,...\,,M\}$ the quantum statistics forces us to further constrain, respectively, locality or free choice. The case of local fraction $\mu_{\scriptscriptstyle L}$ with arbitrary number of settings $M\rightarrow\infty$ has been thoroughly investigated  for statistics generated by quantum states. Let us refer to two interesting results in the literature on local fraction $\mu_{\scriptscriptstyle L}$ which readily translate via \textbf{Theorem~\ref{Theorem-equivalence}} to the measure of freedom $\mu_{\scriptscriptstyle F}$. The first one concerns the statistics of a maximally entangled state, cf.~\cite{ElPoRo92,BaKePi06} (see \textbf{SI Appendix} for a direct proof).
\begin{theorem}\label{Theorem-ChainedBell}
For every local hidden variable (LHV) model that explains the statistics of a Bell experiment for a maximally entangled state the amount of free choice tends to zero with increasing number of measurement settings $M$, i.e. $\mu_{\scriptscriptstyle F}\xymatrix{\ar[r]_{ M\,\rightarrow\,\infty\atop} &}
0\,$\,.\vspace{-0.3cm}
\end{theorem}

Apparently, for less entangled states the amount of freedom increases, reaching the maximal value $\mu_{\scriptscriptstyle F}=1$ for separable states. This is a consequence of the result in~\cite{PoBrGi12}, which explicitly computes the local fraction $\mu_{\scriptscriptstyle L}$ for all pure two-qubit states. Stated for measure $\mu_{\scriptscriptstyle F}$ this takes the following form.
\begin{theorem}\label{Theorem-Portmann}
For a pure two-qubit state, which by appropriate choice of the basis can always be written in the form $\ket{\psi}=\cos\tfrac{\theta}{2}\ket{00}+\sin\tfrac{\theta}{2}\ket{11}$ with $\theta\in[0,\tfrac{\pi}{2}]$, the amount of freedom is equal $\mu_{\scriptscriptstyle F}=\cos\,\theta$, whatever the choice and number of settings on Alice and Bob's side.
\end{theorem}
Note that both \textbf{Theorem~\ref{Theorem-ChainedBell}} and \textbf{Theorem~\ref{Theorem-Portmann}} assume a specific form of behaviour $\{P_{\scriptscriptstyle ab|xy}\}_{\scriptscriptstyle xy}$ as obtained by the rules of quantum theory. The theorems should be contrasted with  \textbf{Theorem~\ref{Theorem-CHSH}} which is a \textit{theory-independent} statement,  not limited to a particular theoretical framework.

\section*{Discussion}

The ingenuity of Bell's theorem lies in the fundamental nature of the premises from which the result is derived. Within the standard causal (or \textit{realist}) approach, it is hard to assume less about two agents than having \textit{free choice} and their systems being \textit{localised} in space. Yet in some experiments nature refutes the possibility that both assumptions are concurrently true~\cite{AsDaRo82,GiVeWeHaHoPhSt15,ShMeChBiWaStGe15,HeBeDrReKaBlRu15,As15,GaFrKa14,RaHaHoGaFrLeLi18,AbAcAlAlAnAnBe18}. It is not easy to reject either one of them without carefully rethinking the role of observers and how cause-and-effect manifests in the world.\footnote{We note that the conventional understanding of causality and the language of counterfactuals has recently gained a solid mathematical basis; see e.g. the work of J.~Pearl~\cite{Pe09}. However, in view of the apparent difficulties with embedding quantum mechanics in that framework, the standard approach to causality based on Reichenbach's principle or claims regarding spatio-temporal structure of events might need reassessment; see e.g. indefinite causal structures~\cite{Br14,AlBaHoLeSp17} or retrocausality~\cite{WhAr20}.} Our objective in this paper is this: \textit{instead of pondering the question of how this could be possible, we ask about the extent to which a given assumption has to be relaxed in order to maintain the other}. Expressed more colloquially, it is natural for a realist to ask what is the cost of trading one concept for the other: \textit{Is it possible to save free choice by giving up on \textit{some} locality? Or, maybe is it better to forego a \textit{modicum} of free choice in exchange for locality?} These questions can be compared on equal footing by computing a proportion of trials across the whole experimental ensemble in which a given assumption must fail, when the other holds at all times. Surprisingly, the answer can be obtained by looking at the observed statistics alone (avoiding the specifics of the experimental setup). The first question was formulated in the quantum-mechanical context  by Elitzur, Popescu and Rohrlich~\cite{ElPoRo92} who introduced the notion of \textit{local fraction} further elaborated in~\cite{BaKePi06,CoRe08,CoRe16,PoBrGi12}  (see~\cite{Ha91} for an early indication of these ideas). Here, we generalise this notion to arbitrary experimental statistics (see also~\cite{AbBaMa17}). Furthermore, we answer the second question by adopting a similar approach to measuring the amount of free choice (which by analogy may be called \textit{free fraction}). The first main result, \textbf{Theorem~\ref{Theorem-equivalence}}, compares such defined measures in the general case (arbitrary statistics with any number of settings), showing that both assumptions are \textit{equally costly}. This demonstrates a deeper symmetry between locality and free choice, which may come as a surprise, given our intuition of a profound difference in the role these concepts play in the description of an experiment. 

In this paper, the notions of locality and free choice are understood in the usual sense required to derive Bell's theorem~\cite{Be87a,CoRe13a}. They are expressed  in the standard causal model framework (which subsumes realism)
as unambiguous yes-no criteria for each experimental trial (i.e. when all past variables are fixed), determining whether there is a causal link between certain variables in a model (without pondering its exact nature). The measures $\mu_{\scriptscriptstyle L}$ and $\mu_{\scriptscriptstyle F}$ count the fraction of trials when such a connection needs to be established, breaking locality or free choice respectively, in order to explain the observed statistics.
This problem is prior to a discussion of how this actually occurs, which is  particularly relevant when the exact nature of the phenomenon under study is obscured.  \textbf{Theorem~\ref{Theorem-equivalence}} shows no intrinsic reason for a realist to favour one assumption vs the other. The minimal frequency of the required causal influences of either sort, measured by  $\mu_{\scriptscriptstyle L}$ and $\mu_{\scriptscriptstyle F}$, is exactly the same. Notably, this is a general result which holds for \textit{any} behaviour $\{P_{\scriptscriptstyle  ab|xy}\}_{\scriptscriptstyle  xy}$. What remains is explicit calculation of those measures for a given experimental statistics. 

The second main result, \textbf{Theorem~\ref{Theorem-CHSH}}, evaluates both measures $\mu_{\scriptscriptstyle L}$ and $\mu_{\scriptscriptstyle F}$ for any \textit{non-signalling} behaviour in a Bell experiment with two outcomes and two settings. It provides a direct interpretation to the amount of violation of the CHSH inequalities~\cite{ClHoShHo69}. The key motivation behind this result is that the degree by which the inequalities are violated has not been given tangible interpretation so far, beyond its use as a binary test of whether the inequalities are obeyed or not in study of Bell non-locality. Furthermore, \textbf{Theorem~\ref{Theorem-CHSH}} has the advantage of being \textit{theory-independent} in the sense of being applicable to the experimental statistics regardless of its theoretical origin (i.e., beyond the quantum-mechanical  framework). This makes it suitable for quantitative assessment of the degree of locality and free choice across different experimental situations, with prospective applications beyond physics, e.g. in neuroscience, cognitive psychology, social sciences or finance~\cite{BuBr12,HaKh13,PoBu13,Ha17a,Ha17}. 

We also state two results,  \textbf{Theorem~\ref{Theorem-ChainedBell}} and \textbf{Theorem~\ref{Theorem-Portmann}}, for the measure of free choice $\mu_{\scriptscriptstyle F}$ in the case of the quantum statistics generated by the pure two-qubit states. Both are direct translation, via \textbf{Theorem~\ref{Theorem-equivalence}}, of the corresponding results for the local fraction $\mu_{\scriptscriptstyle L}$~\cite{ElPoRo92,Ha91,BaKePi06,CoRe08,CoRe16,PoBrGi12}. 

It is worth noting a related idea of quantifying non-locality through the amount of information transmitted between the parties that is required to reproduce quantum correlations (under free choice assumption). Together with the development of the specific models~\cite{Ma92a,BrClTa99,St00,ToBa03,Gi20}, this has led to various results regarding communication complexity in the quantum realm~\cite{BuClMaWo10}. However, in this paper we take a different perspective on measuring non-locality by changing the question from \textit{"how} \textit{much"} to \textit{"how} \textit{often"} communication needs to be established between the parties to simulate given correlations. \textbf{Theorem~\ref{Theorem-CHSH}} gives the exact bound in the case of non-signalling statistics in the two-setting and two-outcome Bell experiment. In the quantum case, such a simulation requires communication in at least 41\,\% of trials (because of Tsirelson's bound~\cite{Ts80}) and for maximally entangled states increases to 100\,\% of trials when the number of settings is arbitrary (cf. \textbf{Theorems~\ref{Theorem-ChainedBell}} and \textbf{\ref{Theorem-Portmann}}).

Natural as it may seem, the idea of measuring freedom of  choice by measure $\mu_{\scriptscriptstyle F}$ has not been developed in the literature. The reason for this omission can be traced to the conceptual and technical issues with handling arbitrariness of the distribution of settings $P_{\scriptscriptstyle xy}$. Those concerns are properly addressed in the present paper with \textbf{Lemma~\ref{max-lemma}}, which considerably simplifies and supports  \textbf{Definition~\ref{definition}}.
We note that various measures have been developed as a means of quantifying freedom of choice (or \textit{measurement independence}, as it is sometimes called). They include maximal distance between distributions~\cite{Ha10,Ha11b}, mutual information~\cite{BaGi11,HaBr20} or measurement dependent locality~\cite{PuRoBaLiGi14,AkTaMaPuThGi15,PuGi16}. Furthermore, some explicit models simulating correlations in a singlet state with various degrees of measurement dependence have been proposed~\cite{Br88,Ha16} and analysed (e.g. see~\cite{HaBr20} for comparison of causal vs retrocausal models). However, these attempts depart from the original understanding of the free choice as introduced by Bell~\cite{Be87a,CoRe13a} (strict independence of choice from anything in the past) in favour of more sophisticated information-theoretic accounts. Notably, the proposed measure of free choice builds on the Bell's original framework assessing the maximal frequency with which such a freedom \textit{can} be retained in a model strictly consistent with locality. It thus benefits from a direct interpretation within the established causal framework of Bell inequalities and has a clear-cut operational meaning.

Regarding \textbf{Theorem~\ref{Theorem-ChainedBell}}, which rules out \textit{any} freedom of choice \textit{so defined}, it is interesting to take an adversarial perspective on the problem of free choice in relation to quantum cryptography and device independent certification~\cite{KoPaBr06,KoHaSePoMaKaSc12}. In this narrative an eavesdropper controls the devices trying to simulate the quantum statistics of a Bell test, which is impossible as long as the parties enjoy freedom of choice. However, any breach of the latter, i.e. control of measurement settings, shifts the balance in favour of the eavesdropper in her malicious task. Taking the view that any causal influence comes with a cost or danger of being uncovered there are two diverging strategies that reduce the cost/risk to be considered: \textit{(a)} resort to the use of control of choice as seldom as possible during the experiment, or \textit{(b)} minimise the intensity of each act of control. \textbf{Theorem~\ref{Theorem-ChainedBell}} completely rules out the first possibility when simulating quantum statistics, i.e., the eavesdropper needs to manipulate both settings on each trial in order to simulate the quantum statistics. The question about the intensity of the control is left open in our discussion, but amenable to information-theoretic methods~\cite{Ha10,Ha16,Ha10,Ha11b,BaGi11,PuRoBaLiGi14,AkTaMaPuThGi15,PuGi16,HaBr20}. This gives additional security criteria for quantum cryptography and device independent certification by forcing the eavesdropper to a more challenging sort of attack (not only can she not miss a trial, but the control has to be subtle enough). 

We remark that the main \textbf{Theorem~\ref{Theorem-equivalence}} readily extends to the case of larger number of parties and outcomes $\{P_{\scriptscriptstyle  abc...|xyz...}\}_{\scriptscriptstyle xyz...}$. This should be also possible for \textbf{Theorem~\ref{Theorem-CHSH}} when characterisation of the local polytope is known, cf.~\cite{SuZa81,GaMe84,ZuBr02,KlCaBiSh08,BaLiMaPiPoRo05,BaPi05,JoMa05}. Yet another valuable avenue for research in that case consists of completing the analysis to include signalling scenarios~\cite{KuDzLa15,DzKuLa15}. As for the quantum case, we considered the simplest Bell-type scenario with two parties involved in the experiment, but extensions may prove even more surprising (see~\cite{Sc19} for a  technical review of the vast field of Bell non-locality). In particular, in three-party scenarios the methods discussed presently can be used to eliminate freedom of choice already for two settings per party sharing the GHZ state (cf. Mermin inequalities which saturate in that case~\cite{Me90a}). 
We should also mention an intriguing result~\cite{ReBaBoBrGiBe19} for a triangle quantum network in which non-locality can be proved with all measurements fixed. Remarkably, there is nothing to choose in that setup, but there is another assumption of preparation independence which plays a crucial role in the argument.

In this paper we are trying to remain impartial as to which assumption --- \textit{locality} or \textit{free choice} --- is more important on the fundamental level. This is certainly a strongly debated subject in general, both  among physicists and philosophers, with strong supporters on each side~\cite{Ma19,La19,No17}. As just one example depreciating the role of freedom of choice let us quote Albert~Einstein\footnote{Statement to the Spinoza Society of America. September 22, 1932. AEA 33-291.}: \textit{"Human beings, in their thinking, feeling and acting are not free agents but are as causally bound as the stars in their motion."} As a counterbalance, it is hard to resist the objection that was eloquently stated by Nicolas~Gisin~\cite{Gi14b}: \textit{"But for me, the situation is very clear: not only does free will exist, but it is a prerequisite for science, philosophy, and our very ability to think rationally in a meaningful way."} Without entering into this debate, we remark that both assumptions are interchangeable on a deeper level. Namely, for a given experimental statistics $\{P_{\scriptscriptstyle ab|xy}\}_{\scriptscriptstyle xy}$ in a Bell-type experiment the measure of locality $\mu_{\scriptscriptstyle L}$ and measure of free choice $\mu_{\scriptscriptstyle F}$ are exactly the same. This makes an even stronger case regarding the inherent impossibility of inferring causal structure from experimental statistics alone.

\matmethods{In order to facilitate the following discussion we begin with two technical lemmas. See \textbf{SI Appendix} for the proofs.

The first one holds for a Bell experiment with arbitrary number of settings $x,y\in\mathfrak{M}=\{1,\,2,\,3,...\,,M\}$.

\begin{lemma}\label{Lemma-dilation}
For any behaviour $\{{P}_{\scriptscriptstyle ab|xy}\}_{\scriptscriptstyle xy}$ and distribution of settings $P_{\scriptscriptstyle xy}$ there exists a local hidden variable model (LHV) which \underline{fully} violates the freedom of choice assumption.
[i.e. if $\tilde{\Lambda}$ is the relevant HV space, then we have $\tilde{\Lambda}=\tilde{\Lambda}_{\scriptscriptstyle L}=\tilde{\Lambda}_{\scriptscriptstyle NF}$, cf. Eqs.~(\ref{locality-splitting}) and (\ref{freedom-splitting}))].
\end{lemma}

The second one concerns a Bell scenario with binary settings $x,y\in\mathfrak{M}=\{0,\,1\}$. 
\begin{lemma}\label{Lemma-decomposition}
Each non-signalling behaviour $\{P_{\scriptscriptstyle ab|xy}\}_{\scriptscriptstyle xy}$ with binary settings $x,y\in\mathfrak{M}=\{0,\,1\}$ can be decomposed as a convex mixture of a local behaviour $\{\bar{P}_{\scriptscriptstyle ab|xy}\}_{\scriptscriptstyle xy}$ and a PR-box $\{\tilde{P}_{\scriptscriptstyle ab|xy}\}_{\scriptscriptstyle xy}$ in the form
\begin{eqnarray}\label{decomposition}
P_{\scriptscriptstyle ab|xy}\ =\ p\cdot \bar{P}_{\scriptscriptstyle ab|xy}+(1-p)\cdot \tilde{P}_{\scriptscriptstyle ab|xy}\,,
\end{eqnarray}
with $p=\tfrac{1}{2}(4-S_{\scriptscriptstyle max})$ for all $x,y\in\{0,1\}$.
\end{lemma}
\noindent Recall that a PR-box~\cite{PoRo94} is a non-signalling behaviour for which one of the CHSH expressions in Eqs.~(\ref{S1})-(\ref{S4}) reaches the maximal algebraic bound of $|S_{\scriptscriptstyle i}|=4$. Here, local behaviour means existence of a LHV+FHV model of $\{\bar{P}_{\scriptscriptstyle ab|xy}\}_{\scriptscriptstyle xy}$ and $S_{\scriptscriptstyle max}=\max\,\{|S_{\scriptscriptstyle i}|:i=1,...\,,4\}$.

We are now ready to proceed with the proofs.

\subsection*{Proof of Lemma~{\ref{max-lemma}}}
Suppose we have a HV model (\ref{HV}) of some behaviour $\{P_{\scriptscriptstyle ab|xy}\}_{\scriptscriptstyle xy}$ for some nontrivial distribution of settings  ${P}_{\scriptscriptstyle xy}$. The latter obtains via Eq.~(\ref{distribution-settings}) from the conditional probabilities ${P}_{\scriptscriptstyle xy|\lambda}$ which are related to probabilities specified by the model, ${P}_{\scriptscriptstyle \lambda|xy}$ and ${P}_{\scriptscriptstyle \lambda}$\,, by the usual Bayes' rule. The point at issue is whether a given HV model can simulate any other distribution of settings $\tilde{P}_{\scriptscriptstyle xy}$ via Eq.~(\ref{distribution-settings}) by changing ${P}_{\scriptscriptstyle xy|\lambda}\leadsto\tilde{P}_{\scriptscriptstyle xy|\lambda}$\,, while keeping the remaining components of the HV model (\ref{HV}) intact. This requires consistency with Bayes' rule, i.e.
\begin{eqnarray}\label{Bayes}
\tilde{P}_{\scriptscriptstyle xy|\lambda}\ =\ \tfrac{{P}_{\scriptscriptstyle \lambda|xy}\cdot\,\tilde{P}_{xy}}{{P}_{\scriptscriptstyle \lambda}}\,,
\end{eqnarray}
which should be a well-defined probability distribution for each $\lambda$. Since distributions ${P}_{\scriptscriptstyle \lambda|xy}$ and ${P}_{\scriptscriptstyle \lambda}$ are fixed by the HV model (\ref{HV}), then the distribution of settings $\tilde{P}_{\scriptscriptstyle xy}$ is arbitrary as long as the expression in Eq.~(\ref{Bayes}) is less then 1 for each $\lambda\in\Lambda$ (normalisation is trivially fulfilled). Now, whenever freedom of choice from Eq.~(\ref{free-choice}) holds, this condition is always satisfied, and hence such a HV model can be trivially adjusted for any distribution $\tilde{P}_{\scriptscriptstyle xy}$ (by redefining $\tilde{P}_{\scriptscriptstyle xy|\lambda}:=\tilde{P}_{\scriptscriptstyle xy}$ in compliance with Eq.~(\ref{Bayes}), and keeping all the remaining components of the HV model (\ref{HV}) unchanged). Of course, for FHV models in the definition of $\mu_{\scriptscriptstyle L}$ in Eq.~(\ref{locality-measure}) this is the case, which thus entails the simpler expression for $\mu_{\scriptscriptstyle L}$ in Eq.~(\ref{locality-measure-max}).

Clearly, such a simple argument falls apart for models without freedom of choice, like those in the definition of $\mu_{\scriptscriptstyle F}$ in Eq.~(\ref{freedom-measure}), when  ${P}_{\scriptscriptstyle \lambda|xy}$ and ${P}_{\scriptscriptstyle \lambda}$ \textit{do not} cancel out and the probability in Eq.~(\ref{Bayes}) may be ill-defined. In that case, some deeper intervention into the model is required as shown below. 

Let us take some LHV model (\ref{HV}) simulating a given behaviour $\{P_{\scriptscriptstyle ab|xy}\}_{\scriptscriptstyle xy}$ with nontrivial distribution of settings $P_{\scriptscriptstyle xy}$. Then the related HV space decomposes as $\Lambda=\Lambda_{\scriptscriptstyle F}\uplus\Lambda_{\scriptscriptstyle NF}$ and the degree of freedom is measured by $p_{\scriptscriptstyle F}:=\sum_{\scriptscriptstyle \lambda\in\Lambda_{F}}P_{\scriptscriptstyle \lambda}$\,, cf. \textbf{Remark~\ref{remark-definition}}. Now, consider a restriction of the model to the respective subspaces $\Lambda_{\scriptscriptstyle F}$ and $\Lambda_{\scriptscriptstyle NF}$ which amounts to the following rescaling 
\begin{eqnarray}
P_{\scriptscriptstyle \lambda}^{\scriptscriptstyle F}\ :=\ \tfrac{1}{p_{\scriptscriptstyle F}}\,P_{\scriptscriptstyle \lambda}\,,\quad P_{\scriptscriptstyle \lambda|xy}^{\scriptscriptstyle F}\ :=\ \tfrac{1}{p_{\scriptscriptstyle F}}\,P_{\scriptscriptstyle \lambda|xy}\,,\quad P_{\scriptscriptstyle ab|xy\lambda}^{\scriptscriptstyle F}\ :=\ P_{\scriptscriptstyle ab|xy\lambda}\,,
\end{eqnarray}
for $\lambda\in\Lambda_{\scriptscriptstyle F}$, and similarly
\begin{eqnarray}
P_{\scriptscriptstyle \lambda}^{\scriptscriptstyle NF}\ :=\ \tfrac{1}{1-p_{\scriptscriptstyle F}}\,P_{\scriptscriptstyle \lambda}\,,\quad P_{\scriptscriptstyle \lambda|xy}^{\scriptscriptstyle NF}\ :=\ \tfrac{1}{1-p_{\scriptscriptstyle F}}\,P_{\scriptscriptstyle \lambda|xy}\,,\quad P_{\scriptscriptstyle ab|xy\lambda}^{\scriptscriptstyle NF}\ :=\ P_{\scriptscriptstyle ab|xy\lambda}\,,
\end{eqnarray}
for $\lambda\in\Lambda_{\scriptscriptstyle NF}$. Both are LHV models with marginals 
\begin{eqnarray}\label{Pab|xy-F}
P_{\scriptscriptstyle ab|xy}^{\scriptscriptstyle F}\ =\ \sum_{\scriptscriptstyle \lambda\in\Lambda_{\scriptscriptstyle F}}\,P_{\scriptscriptstyle ab|xy\lambda}^{\scriptscriptstyle F}\cdot P_{\scriptscriptstyle \lambda|xy}^{\scriptscriptstyle F}\,,
\\\label{Pab|xy-NF}
P_{\scriptscriptstyle ab|xy}^{\scriptscriptstyle NF}\ =\ \sum_{\scriptscriptstyle \lambda\in\Lambda_{\scriptscriptstyle NF}}\!P_{\scriptscriptstyle ab|xy\lambda}^{\scriptscriptstyle NF}\cdot P_{\scriptscriptstyle \lambda|xy}^{\scriptscriptstyle NF}\,,
\end{eqnarray}
which provide a convex decomposition of the original behaviour  $\{P_{\scriptscriptstyle ab|xy}\}_{\scriptscriptstyle xy}$\,, i.e.
\begin{eqnarray}\label{Pab|xy-convex}
P_{\scriptscriptstyle ab|xy}\ =\ p_{\scriptscriptstyle F}\cdot P_{\scriptscriptstyle ab|xy}^{\scriptscriptstyle F}+(1-p_{\scriptscriptstyle F})\cdot P_{\scriptscriptstyle ab|xy}^{\scriptscriptstyle NF}\,.
\end{eqnarray}

The crucial point is a  careful adjustment of these two models to recover some arbitrary distribution of settings $\tilde{P}_{\scriptscriptstyle xy}$, while maintaining the respective marginals Eqs.~(\ref{Pab|xy-F}) and (\ref{Pab|xy-NF}). For the first one (restriction to $\Lambda_{\scriptscriptstyle F}$) the situation is trivial as explained above: since it is a FHV model, then it suffice to redefine $\tilde{P}_{\scriptscriptstyle xy|\lambda}^{\scriptscriptstyle F}:=\tilde{P}_{\scriptscriptstyle xy}$ (in compliance with Eq.~(\ref{Bayes})) and leave all rest intact. As for the second one (restriction to $\Lambda_{\scriptscriptstyle NF}$), we can use \textbf{Lemma~\ref{Lemma-dilation}} for constructing another HV space $\tilde{\Lambda}_{\scriptscriptstyle NF}$ with a LHV model without any free choice, that simulates behaviour $\{P_{\scriptscriptstyle ab|xy}^{\scriptscriptstyle NF}\}_{\scriptscriptstyle xy}$ with the required distribution of settings $\tilde{P}_{\scriptscriptstyle xy}$. Then, such modified models can be stitched back together on the compound HV space $\tilde{\Lambda}:=\Lambda_{\scriptscriptstyle F}\uplus\tilde{\Lambda}_{\scriptscriptstyle NF}$ with respective weights $p_{\scriptscriptstyle F}$ and $1-p_{\scriptscriptstyle F}$. This guarantees reconstruction of the original behaviour  $\{P_{\scriptscriptstyle ab|xy}\}_{\scriptscriptstyle xy}$ (see Eq.~(\ref{Pab|xy-convex})) with the new distribution of settings $\tilde{P}_{\scriptscriptstyle xy}$. The model is local and has the same degree of freedom equal to $p_{\scriptscriptstyle F}$ (the first component has full freedom of choice, while in the second one it is entirely missing).

The above construction shows that for every LHV model of some behaviour $\{P_{\scriptscriptstyle ab|xy}\}_{\scriptscriptstyle xy}$ there is always another one adjusted for any other distribution of settings $\tilde{P}_{\scriptscriptstyle xy}$ with the same degree of freedom. This justifies the simpler expression for $\mu_{\scriptscriptstyle F}$ in Eq.~(\ref{freedom-measure-max}) and hence concludes the proof of \textbf{Lemma~\ref{max-lemma}}.

\subsection*{Proof of Theorem~{\ref{Theorem-equivalence}}}
Note that \textbf{Lemma~\ref{max-lemma}} Eqs.~(\ref{locality-measure-max}) and (\ref{freedom-measure-max}) can be taken as a definition of measures $\mu_{\scriptscriptstyle L}$ and $\mu_{\scriptscriptstyle F}$. This is very convenient, since it allows a discussion free from any concerns about the distribution of settings $P_{xy}$ (this is particularly relevant in the case of $\mu_{\scriptscriptstyle F}$ as explained above).

It is instructive to observe that the calculation of both measures $\mu_{\scriptscriptstyle L}$ and $\mu_{\scriptscriptstyle F}$ can be succinctly formulated as a convex optimisation problem. Suppose, we can decompose some given behaviour $\{P_{\scriptscriptstyle ab|xy}\}_{\scriptscriptstyle xy}$ as a mixture 
\begin{eqnarray}\label{convex-locality}
P_{\scriptscriptstyle ab|xy}\ =\ p_{\scriptscriptstyle L}\cdot P_{\scriptscriptstyle ab|xy}^{\scriptscriptstyle L}+(1-p_{\scriptscriptstyle L})\cdot P_{\scriptscriptstyle ab|xy}^{\scriptscriptstyle NL}\,,
\end{eqnarray}
where $\{P_{\scriptscriptstyle ab|xy}^{\scriptscriptstyle L}\}_{\scriptscriptstyle xy}$ is a local behaviour with full freedom of choice (i.e., has a LHV+FHV model), and $\{P_{\scriptscriptstyle ab|xy}^{\scriptscriptstyle NL}\}_{\scriptscriptstyle xy}$ is a free behaviour (i.e., has a FHV model). And similarly, suppose that
\begin{eqnarray}\label{convex-freedom}
P_{\scriptscriptstyle ab|xy}\ =\ p_{\scriptscriptstyle F}\cdot P_{\scriptscriptstyle ab|xy}^{\scriptscriptstyle F}+(1-p_{\scriptscriptstyle F})\cdot P_{\scriptscriptstyle ab|xy}^{\scriptscriptstyle NF}\,
\end{eqnarray}
where $\{P_{\scriptscriptstyle ab|xy}^{\scriptscriptstyle F}\}_{\scriptscriptstyle xy}$ is a local behaviour with full freedom of choice (i.e., has a LHV+FHV model), and $\{P_{\scriptscriptstyle ab|xy}^{\scriptscriptstyle NF}\}_{\scriptscriptstyle xy}$ is a local behaviour (i.e., has a LHV model). In both cases we assume that $0\leqslant p_{\scriptscriptstyle L}\,,p_{\scriptscriptstyle F}\leqslant1$, and both Eq.~(\ref{convex-locality}) and Eq.~(\ref{convex-freedom}) have to hold for all $a,b=\pm1$ and $x,y\in\mathfrak{M}$. Then, we have 

\begin{remark}\label{remark-fraction}
Measures $\mu_{\scriptscriptstyle L}$ and $\mu_{\scriptscriptstyle F}$ evaluate the maxima over all possible decompositions in Eqs.~(\ref{convex-locality}) and (\ref{convex-freedom}) of behaviour $\{P_{\scriptscriptstyle ab|xy}\}_{\scriptscriptstyle xy}$, i.e. 
\begin{eqnarray}\label{max-convex-locality}
\mu_{\scriptscriptstyle L}\ =\ \max_{\scriptscriptstyle \text{decomp. (\ref{convex-locality})}}\ p_{\scriptscriptstyle L}\,,\\\label{max-convex-freedom}
\mu_{\scriptscriptstyle F}\ =\ \max_{\scriptscriptstyle \text{decomp. (\ref{convex-freedom})}}\ p_{\scriptscriptstyle F}\,.
\end{eqnarray}
\end{remark}
\begin{proof} We will justify only Eq.~(\ref{max-convex-locality}), since the argument for Eq.~(\ref{max-convex-freedom}) is analogous.

Let us observe that every HV model (\ref{HV}) of behaviour $\{P_{\scriptscriptstyle ab|xy}\}_{\scriptscriptstyle xy}$ as described by Eq.~(\ref{Pab|xy}) splits into two components (cf. Eq.~(\ref{locality-splitting}))
\begin{eqnarray}\label{P=Non-Local}
P_{\scriptscriptstyle ab|xy}\ =\ \underbrace{\sum_{\scriptscriptstyle \lambda\in\Lambda_{\scriptscriptstyle L}}P_{\scriptscriptstyle ab|xy\lambda}\cdot P_{\scriptscriptstyle \lambda}}_{p_{\scriptscriptstyle L}\,\cdot\, P_{\scriptscriptstyle ab|xy}^{\scriptscriptstyle L}}\ +\underbrace{\sum_{\scriptscriptstyle \lambda\in\Lambda_{\scriptscriptstyle NL}}P_{\scriptscriptstyle ab|xy\lambda}\cdot P_{\scriptscriptstyle \lambda|xy}}_{(1-p_{\scriptscriptstyle L})\,\cdot\, P_{\scriptscriptstyle ab|xy}^{\scriptscriptstyle NL}}\,\,,
\end{eqnarray}
which defines decomposition of the type in Eq.~(\ref{convex-locality}) with $p_{\scriptscriptstyle L}:=\sum_{\scriptscriptstyle \lambda\in\Lambda_{\scriptscriptstyle L}}P_{\scriptscriptstyle \lambda}$. Therefore, by Eq.~(\ref{locality-measure-max}), we get $\mu_{\scriptscriptstyle L}\leqslant{\max}_{\scriptscriptstyle \text{decomp. (\ref{convex-locality})}}\ p_{\scriptscriptstyle L}$.

To see the reverse, we note that every decomposition of the type in Eq.~(\ref{convex-locality}) implies  existence of a LHV+FHV model of behaviour $\{P_{\scriptscriptstyle ab|xy}^{\scriptscriptstyle L}\}_{\scriptscriptstyle xy}$ on some HV space $\tilde{\Lambda}_{\scriptscriptstyle L}$ and a FHV model of behaviour $\{P_{\scriptscriptstyle ab|xy}^{\scriptscriptstyle NL}\}_{\scriptscriptstyle xy}$ on some HV space $\tilde{\Lambda}_{\scriptscriptstyle NL}$. Those two models, when combined on a compound HV space 
$\Lambda:=\tilde{\Lambda}_{\scriptscriptstyle L}\uplus\tilde{\Lambda}_{\scriptscriptstyle NL}$
with the respective weights $p_{\scriptscriptstyle L}$ and $1-p_{\scriptscriptstyle L}$\,, provide a HV model of behaviour $\{P_{\scriptscriptstyle ab|xy}\}_{\scriptscriptstyle xy}$. Since the local domain of such a model contains $\tilde{\Lambda}_{\scriptscriptstyle L}$, then from Eq.~(\ref{locality-measure-max}) we have $\mu_{\scriptscriptstyle L}\geqslant p_{\scriptscriptstyle L}$\,, which entails $\mu_{\scriptscriptstyle L}\geqslant{\max}_{\scriptscriptstyle \text{decomp. (\ref{convex-locality})}}\ p_{\scriptscriptstyle L}$. This concludes the proof of Eq.~(\ref{max-convex-locality}).
\end{proof}

Now, in order to prove \textbf{Theorem~{\ref{Theorem-equivalence}}} it is enough to show that for every decomposition of the type in Eq.~(\ref{convex-locality}) there exists a decomposition of the type in Eq.~(\ref{convex-freedom}) with the same weight $p_{\scriptscriptstyle L}=p_{\scriptscriptstyle F}$, and vice versa. A closer look at both expressions reveals that behaviours $\{P_{\scriptscriptstyle ab|xy}^{\scriptscriptstyle L}\}_{\scriptscriptstyle xy}$ and $\{P_{\scriptscriptstyle ab|xy}^{\scriptscriptstyle F}\}_{\scriptscriptstyle xy}$ are both local with full freedom of choice (i.e., share the same  LHV+FHV model). Thus, the problem can be reduced to justifying that: \textit{(a)} behaviour $\{P_{\scriptscriptstyle ab|xy}^{\scriptscriptstyle NL}\}_{\scriptscriptstyle xy}$ also has a LHV model (possibly a non-FHV model), and \textit{(b)} behaviour $\{P_{\scriptscriptstyle ab|xy}^{\scriptscriptstyle NF}\}_{\scriptscriptstyle xy}$ also has a FHV model (possibly a non-LHV model).

\noindent\textit{Ad.~(a)}~This readily follows from \textbf{Lemma~\ref{Lemma-dilation}}.

\noindent\textit{Ad.~(b)}~Here, a trivial model will suffice. Let us take $\Lambda:=\{\lambda_o\}$ (a single-element set) with $P_{\scriptscriptstyle {\lambda_o}}\equiv P_{\scriptscriptstyle \lambda_o|xy}:=1$ and conditional distribution defined as $P_{\scriptscriptstyle ab|xy\lambda_o}:=P_{\scriptscriptstyle ab|xy}^{\scriptscriptstyle NF}$. Clearly, it is a FVH model of behaviour $\{P_{\scriptscriptstyle ab|xy}^{\scriptscriptstyle NF}\}_{\scriptscriptstyle xy}$.

Thus, we have shown equivalence of both decompositions Eqs.~(\ref{convex-locality}) and (\ref{convex-freedom}), which, by virtue of \textbf{Remark~\ref{remark-fraction}}, proves \textbf{Theorem~{\ref{Theorem-equivalence}}}.

\subsection*{Proof of Theorem~{\ref{Theorem-CHSH}}}
By virtue of \textbf{Theorem~{\ref{Theorem-equivalence}}} it suffices to prove the result for one of the measures. Let it be measure $\mu_{\scriptscriptstyle L}$ evaluated by means of Eq.~(\ref{max-convex-locality}) in \textbf{Remark~{\ref{remark-fraction}}}.

Consider some arbitrary decomposition Eq.~(\ref{convex-locality}) of behaviour $\{P_{\scriptscriptstyle ab|xy}\}_{\scriptscriptstyle xy}$. Then, by linearity, the four CHSH expressions Eqs.~(\ref{S1})-(\ref{S4}) decompose as well, i.e. we get
\begin{eqnarray}\label{convex-S}
S_{\scriptscriptstyle  i}\ =\ p_{\scriptscriptstyle L}\cdot S_{\scriptscriptstyle  i}^{\scriptscriptstyle L}+(1-p_{\scriptscriptstyle L})\cdot S_{\scriptscriptstyle  i}^{\scriptscriptstyle NL}\,,
\end{eqnarray}
where $S_{\scriptscriptstyle  i}^{\scriptscriptstyle L}$ and $S_{\scriptscriptstyle  i}^{\scriptscriptstyle NL}$ are calculated for the respective behaviours $\{P_{\scriptscriptstyle ab|xy}^{\scriptscriptstyle L}\}_{\scriptscriptstyle xy}$ and $\{P_{\scriptscriptstyle ab|xy}^{\scriptscriptstyle NL}\}_{\scriptscriptstyle xy}$. Since the first one is a local behaviour with full freedom of choice (i.e. having a LHV+FHV model), then from the CHSH inequalities Eq.~(\ref{Bell-CHSH-inequalities}) we have $|S^{\scriptscriptstyle L}_{\scriptscriptstyle i}|\ \leqslant\ 2$. For the second one there is nothing interesting to be said other than noting  the maximal algebraic bound $|S^{\scriptscriptstyle NL}_{\scriptscriptstyle i}|\ \leqslant4$. 
As a consequence, the following inequality obtains
\begin{eqnarray}
|S_{\scriptscriptstyle i}|\ \leqslant\ p_{\scriptscriptstyle L}\cdot 2\,+\,(1-p_{\scriptscriptstyle L})\cdot 4\ =\ 4-2\,p_{\scriptscriptstyle L}\,,
\end{eqnarray}
and we get $p_{\scriptscriptstyle L}\leqslant\tfrac{1}{2}(4-|S_{\scriptscriptstyle i}|)$. Thus, by assumed arbitrariness of decomposition, Eq.~(\ref{convex-locality}) gives the upper bound on expression in Eq.~(\ref{max-convex-locality})
\begin{eqnarray}
\mu_{\scriptscriptstyle{L}}\ \leqslant\ \tfrac{1}{2}(4-|S_{\scriptscriptstyle i}|)\,,
\end{eqnarray}
where  $S_{\scriptscriptstyle max}=\max\,\{|S_{\scriptscriptstyle i}|:i=1,...\,,4\}$. 
By \textbf{Lemma~\ref{Lemma-decomposition}} we conclude that the bound is tight, which ends the proof of \textbf{Theorem~\ref{Theorem-CHSH}}.

}

\showmatmethods{} 

\acknow{We thank R.~Colbeck, J.~Duda, N.~Gisin, M.~Hall, L.~Hardy, D.~Kaiser, M.~Markiewicz, S.~Pironio, D.~Rohrlich and V.~Scarani for helpful comments. PB acknowledges support from the Polish National Agency for Academic Exchange in the Bekker Scholarship Programme. PB and EMP were supported by Office of Naval Research Global grant N62909-19-1-2000.}

\showacknow{} 

\bibliography{CombQuant}

\end{document}




\maketitle

\twocolumn

\SItext

\setcounter{equation}{33}

In this Supplement we prove two technical results \textbf{Lemma~\ref{Lemma-dilation}} and \textbf{Lemma~\ref{Lemma-decomposition}} from the \textbf{Methods} section in the main manuscript. We also give an alternative self-standing proof of \textbf{Theorem~\ref{Theorem-ChainedBell}}. 
\vspace{0.2cm}

\noindent\textit{[The numbering of equations follows the main text.]}
\vspace{0.4cm}

\section*{PROOF OF LEMMA~\ref{Lemma-dilation}}
For reference we quote \textbf{Lemma~\ref{Lemma-dilation}} from the main manuscript that we will shortly prove. Here, the choice of settings is arbitrary $x,y\in\mathfrak{M}=\{1,2\,,\dots,M\}$.
\setcounter{lemma}{1}
\begin{lemma}\label{Lemma-dilation-Supplement}
For any behaviour $\{{P}_{\scriptscriptstyle ab|xy}\}_{\scriptscriptstyle xy}$ and distribution of settings $P_{\scriptscriptstyle xy}$ there exists a local hidden variable model (LHV) which \underline{fully} violates the freedom of choice assumption (i.e., if $\tilde{\Lambda}$ is the relevant HV space, then we have $\tilde{\Lambda}=\tilde{\Lambda}_{\scriptscriptstyle L}=\tilde{\Lambda}_{\scriptscriptstyle NF}$\,; cf. Eqs.~(\ref{locality-splitting}) and (\ref{freedom-splitting})).
\end{lemma}
\begin{proof}
Let us first introduce an auxiliary pair of indices  $uv\in\mathfrak{M}\times\mathfrak{M}$ where $\mathfrak{M}=\{1,2\,,\dots,M\}$ and for each pair $uv$ define a trivial behaviour $\{\tilde{P}^{\scriptscriptstyle (uv)}_{\scriptscriptstyle ab|xy}\}_{\scriptscriptstyle xy}$ which consists of $M\times M$ copies of the same distribution 
\begin{eqnarray}\label{P-trivial}
\tilde{P}^{\scriptscriptstyle (uv)}_{\scriptscriptstyle ab|xy}\ :=\ {P}_{\scriptscriptstyle ab|uv}\quad\text{for all $xy\,\in\mathfrak{M}\times\mathfrak{M}$}\,,
\end{eqnarray}
i.e., this holds independently of the choice of settings $xy$. Then, for each such behaviour there exists an LHV model. To see this, for a given pair of indices $uv$ it is enough to consider a four-element HV space
\begin{eqnarray}
\tilde{\Lambda}^{\!\scriptscriptstyle (uv)}\ :=\ \{(\alpha,\beta):\alpha,\beta=\pm\,1\}\,,
\end{eqnarray}
with probability distributions defined in the following way
\begin{eqnarray}\label{cond}
\tilde{P}^{\scriptscriptstyle (uv)}_{\scriptscriptstyle\lambda}\ :=\ \tilde{P}_{\scriptscriptstyle \alpha\beta|uv}&\ \text{and}\ &\tilde{P}^{\scriptscriptstyle (uv)}_{\scriptscriptstyle ab|xy\lambda}\ :=\ \delta_{\scriptscriptstyle a=\alpha}\cdot\,\delta_{\scriptscriptstyle b=\beta}\,,
\end{eqnarray}
where $\lambda=(\alpha,\beta)\in\tilde{\Lambda}^{\!\scriptscriptstyle (uv)}$. By construction, the locality condition Eq.~(\ref{factorisation}) is satisfied for all $\lambda$, i.e., $\tilde{\Lambda}^{\!\scriptscriptstyle (uv)}=\tilde{\Lambda}^{\!\scriptscriptstyle (uv)}_{\scriptscriptstyle L}$, and the marginals reproduce the behaviour in Eq.~(\ref{P-trivial})
\begin{eqnarray}\label{marginal-free}
\tilde{P}^{\scriptscriptstyle (uv)}_{\scriptscriptstyle ab|xy}\ =\ \sum_{\scriptscriptstyle\lambda\in\tilde{\Lambda}^{\!\scriptscriptstyle (uv)}}\tilde{P}^{\scriptscriptstyle (uv)}_{\scriptscriptstyle ab|xy\lambda}\cdot \tilde{P}^{\scriptscriptstyle (uv)}_{\scriptscriptstyle\lambda}\,.
\end{eqnarray}

Those LVH models of auxiliary behaviours $\{\tilde{P}^{\scriptscriptstyle (uv)}_{\scriptscriptstyle ab|xy}\}_{\scriptscriptstyle xy}$ 
can be used to construct another local model on a compound HV space
\begin{eqnarray}\label{Lambda}
\tilde{\Lambda}\ :=\ \biguplus_{\scriptscriptstyle u,v\,\in\,\mathfrak{M}}\tilde{\Lambda}^{\!\scriptscriptstyle (uv)}\,,
\end{eqnarray}
with conditional probabilities on each component left unchanged, i.e., for each $u,v\in\mathfrak{M}$ we have
\begin{eqnarray}\label{Puv}
\tilde{P}_{\scriptscriptstyle ab|xy\lambda}\ :=\ \tilde{P}^{\scriptscriptstyle (uv)}_{\scriptscriptstyle ab|xy\lambda}\qquad \text{for \ $\lambda\in\tilde{\Lambda}^{\!\scriptscriptstyle (uv)}$\,.}
\end{eqnarray}
Note that in this way the \textit{locality} condition is \underline{retained}, i.e., the factorisation of Eq.~(\ref{factorisation}) holds for every $ \lambda\in\tilde{\Lambda}$ (cf. Eq.~(\ref{cond})), and we have $\tilde{\Lambda}_{\scriptscriptstyle L}=\biguplus_{\scriptscriptstyle u,v\,\in\,\mathfrak{M}}\tilde{\Lambda}^{\!\scriptscriptstyle (uv)}_{\scriptscriptstyle L}=\tilde\Lambda$. Then we impose the conditional distributions $\tilde{P}_{\scriptscriptstyle xy|\lambda}$ in the form
\begin{eqnarray}\label{free-choice-lost}
\tilde{P}_{\scriptscriptstyle xy|\lambda}\ :=\ \delta_{\scriptscriptstyle x=u}\cdot\, \delta_{\scriptscriptstyle y=v}\qquad \text{for \ $\lambda\in\tilde{\Lambda}^{\!\scriptscriptstyle (uv)}$\,.}
\end{eqnarray}
Clearly, this \underline{violates} the \textit{free choice} assumption Eq.~(\ref{free-choice}) for every $ \lambda\in\tilde{\Lambda}$ (since the knowledge of $\lambda$ fully determines the settings $xy$), i.e., we have $\tilde{\Lambda}_{\scriptscriptstyle NF}=\tilde{\Lambda}$. Finally, we define the distribution of $\tilde{P}_{\scriptscriptstyle\lambda}$ on the whole HV space $\tilde{\Lambda}$ by the following condition
\begin{eqnarray}\label{Plambda}
\tilde{P}_{\scriptscriptstyle\lambda}\ :=\ \tilde{P}^{\scriptscriptstyle (uv)}_{\scriptscriptstyle\lambda}\cdot\,P_{\scriptscriptstyle uv}\qquad \text{for \ $\lambda\in\tilde{\Lambda}^{\!\scriptscriptstyle (uv)}$\,,}
\end{eqnarray}
where $P_{\scriptscriptstyle uv}$ is the desired distribution of settings (which for sake of generality is left unspecified). Observe that Eqs.~(\ref{free-choice-lost}) and (\ref{Plambda}) allow us to reverse the conditioning, i.e., Bayes' rule gives
\begin{eqnarray}\label{free-choice-lost-reverse}
\tilde{P}_{\scriptscriptstyle\lambda|xy}\ :=\ \left\{\begin{array}{lll}\tilde{P}^{\scriptscriptstyle (xy)}_{\scriptscriptstyle\lambda}\ &&\text{if \ $\lambda\in\tilde{\Lambda}^{\!\scriptscriptstyle (xy)}$\,,}\\
0&&\text{otherwise\,.}\end{array}\right.
\end{eqnarray}

Having defined in Eqs.~(\ref{Lambda})-(\ref{Plambda}) all components of the HV model (\textit{local} and \underline{fully} violating \textit{freedom of choice}) we need to check whether it correctly reconstructs some given behaviour $\{{P}_{\scriptscriptstyle ab|xy}\}_{\scriptscriptstyle xy}$ and distribution of settings $P_{\scriptscriptstyle xy}$ via Eqs.~(\ref{Pab|xy}) and (\ref{distribution-settings}). The following calculation 
\begin{eqnarray}
\sum_{\scriptscriptstyle\lambda\in\tilde{\Lambda}}\tilde{P}_{\scriptscriptstyle ab|xy\lambda}\cdot\tilde{P}_{\scriptscriptstyle \lambda|xy}&\!\!\stackrel{\scriptscriptstyle (\ref{free-choice-lost-reverse})}{=}\!\!&
\sum_{\scriptscriptstyle\lambda\in\tilde{\Lambda}^{(xy)}}\tilde{P}_{\scriptscriptstyle ab|xy\lambda}\cdot\tilde{P}^{\scriptscriptstyle(xy)}_{\scriptscriptstyle \lambda}\\
&\!\!\stackrel{\scriptscriptstyle (\ref{Puv})}{=}\!\!&\sum_{\scriptscriptstyle\lambda\in\tilde{\Lambda}^{(xy)}}\tilde{P}^{\scriptscriptstyle (xy)}_{\scriptscriptstyle ab|xy\lambda}\cdot\tilde{P}^{\scriptscriptstyle(xy)}_{\scriptscriptstyle \lambda}\\
&\!\!\stackrel{\scriptscriptstyle (\ref{marginal-free})}{=}\!\!&\ \tilde{P}^{\scriptscriptstyle(xy)}_{\scriptscriptstyle ab|xy}
\ \, \stackrel{\scriptscriptstyle (\ref{P-trivial})}{=}\ \,{P}_{\scriptscriptstyle ab|xy}\,,
\end{eqnarray}
confirms the validity of Eq.~(\ref{Pab|xy}). Similarly, we check that 
\begin{eqnarray}
\sum_{\scriptscriptstyle\lambda\in\tilde{\Lambda}}\tilde{P}_{\scriptscriptstyle xy|\lambda}\cdot\tilde{P}_{\scriptscriptstyle \lambda}&\!\!\stackrel{\scriptscriptstyle (\ref{free-choice-lost})}{=}\!\!&\sum_{\scriptscriptstyle\lambda\in\tilde{\Lambda}^{(xy)}}\tilde{P}_{\scriptscriptstyle \lambda}\\
&\!\!\stackrel{\scriptscriptstyle (\ref{Plambda})}{=}\!\!&\sum_{\scriptscriptstyle\lambda\in\tilde{\Lambda}^{(xy)}}\tilde{P}^{\scriptscriptstyle (xy)}_{\scriptscriptstyle\lambda}\cdot P_{\scriptscriptstyle xy}\ \,=\ \,{P}_{\scriptscriptstyle xy}\,,
\end{eqnarray}
in compliance with Eq.~(\ref{distribution-settings}). This concludes the proof.

\end{proof}

\section*{PROOF OF LEMMA~\ref{Lemma-decomposition}}

In the following we consider a Bell experiment with binary settings $x,y\in\mathfrak{M}=\{0\,,1\}$.

\subsection*{Preliminaries}
We begin with a useful observation about the CHSH expressions in Eqs.~(\ref{S1})-(\ref{S4}).
\setcounter{remark}{2}
\begin{remark}\label{Lemma-inequality}
We have $|S_{\scriptscriptstyle i}|+|S_{\scriptscriptstyle j}|\,\leqslant\,4$ whenever $i\neq j$.
\end{remark}
\begin{proof}
This is a direct consequence of a simple algebraic inequality involving the four expressions $S_{\scriptscriptstyle 1},...\,,S_{\scriptscriptstyle 4}$ which follow the pattern of Eqs.~(\ref{S1})-(\ref{S4}), i.e., $S_{\scriptscriptstyle 1}=a+b+c-d$, $S_{\scriptscriptstyle 2}=a+b-c+d$, $S_{\scriptscriptstyle 3}=a-b+c+d$ and $S_{\scriptscriptstyle 4}=-a+b+c+d$. It is straightforward to check that for  $-1\leqslant a,b,c,d\leqslant1$, whenever $i\!\neq\!j$ we have $\pm\,S_{\scriptscriptstyle i}\pm\,S_{\scriptscriptstyle j}\leqslant4$ and $\pm\,S_{\scriptscriptstyle i}\mp S_{\scriptscriptstyle j}\leqslant4$ (since in each case four out of eight terms cancel out). Clearly, in compact form this is equivalent to the inequality $|S_{\scriptscriptstyle i}|+|S_{\scriptscriptstyle j}|\leqslant4$ which concludes the proof.
\end{proof}

\noindent As an immediate consequence of \textbf{Remark~\ref{Lemma-inequality}}, note that \textit{for any behaviour $\{P_{\scriptscriptstyle ab|xy}\}_{\scriptscriptstyle xy}$ at most one out of four CHSH inequalities in Eq.~(\ref{Bell-CHSH-inequalities}) can be violated at the same time}.

Without loss of generality, in the following proof we consider a given behaviour $\{P_{\scriptscriptstyle ab|xy}\}_{\scriptscriptstyle xy}$ with binary choice of settings $x,y\in\mathfrak{M}=\{0,1\}$ for which only the first CHSH inequality is violated, i.e., we assume that
\begin{eqnarray}\label{assumption}
S_{\scriptscriptstyle max}\ =\ S_{\scriptscriptstyle 1}\ >\ 2\,,
\end{eqnarray}
and consequently $|S_{\scriptscriptstyle i}|<2$ for $i=2,3,4$\,. It is straightforward to appreciate that all other cases reduce to this one by relabelling the indices.\footnote{To change signs: $(a\rightarrow-a)\Rightarrow(S_{\scriptscriptstyle i}\rightarrow - S_{\scriptscriptstyle i})$. To switch between expressions: $(x\rightarrow1-x)\Rightarrow (S_{\scriptscriptstyle 1}\leftrightarrow S_{\scriptscriptstyle 3})$, $(y\rightarrow1-y)\Rightarrow (S_{\scriptscriptstyle 1}\leftrightarrow S_{\scriptscriptstyle 2})$ and $(x\rightarrow1-x\ \&\ y\rightarrow1-y)\Rightarrow (S_{\scriptscriptstyle 1}\leftrightarrow S_{\scriptscriptstyle 4})$.}

\subsection*{Decomposition of Lemma~\ref{Lemma-decomposition} (proof)}

\begin{figure}[t]
\centering
\includegraphics[width=\columnwidth]{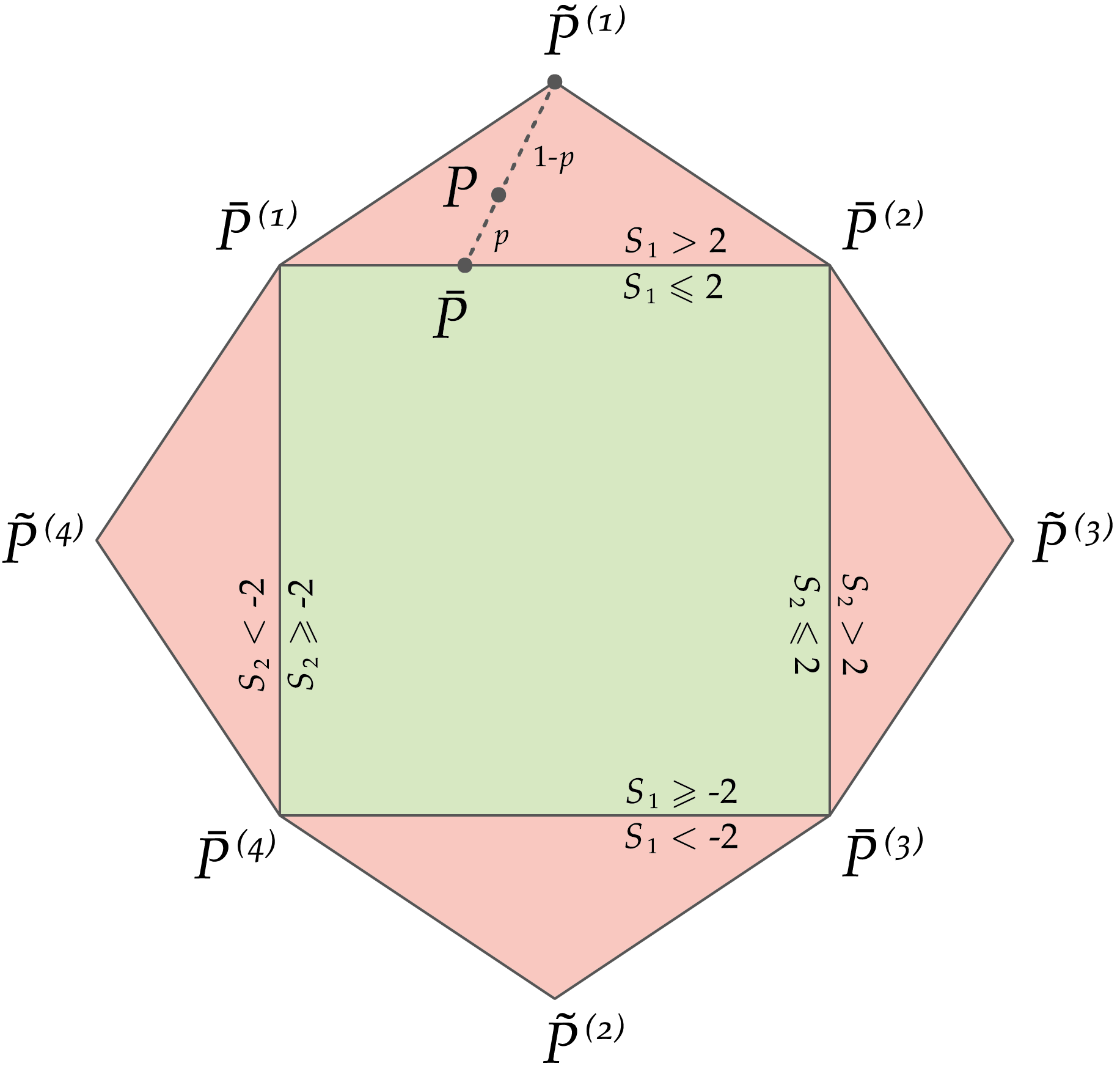}
\caption
{\label{Fig_Decomposition}{\bf\textsf{Decomposition in {\textbf Lemma~\ref{Lemma-decomposition}}.}} A two-dimensional section of the \textit{non-signalling} polytope for the two-setting and two-outcome Bell scenario. It is spanned by sets of \textit{local} deterministic behaviours $\bar{P}^{\scriptscriptstyle\,\textit{(j)}}$ and \textit{non-local} PR-boxes $\tilde{P}^{\scriptscriptstyle\,\textit{(k)}}$. The subset of \textit{local behaviours} (depicted as the inner green rectangle) is fully characterised by the CHSH inequalities Eq.~(\ref{Bell-CHSH-inequalities}). For each non-signalling behaviour $P$ which breaks the local bound, one finds a convex decomposition $P=p\cdot\bar{P}+(1-p)\cdot\tilde{P}^{\scriptscriptstyle\,\textit{(k)}}$ with $p=\tfrac{1}{2}(4-S_{\scriptscriptstyle max})$ for some local behaviour $\bar{P}$ and a PR-box $\tilde{P}^{\scriptscriptstyle\,\textit{(k)}}$. On the picture, we have $S_{\scriptscriptstyle max}=S_{\scriptscriptstyle 1}>2$ and $k=1$.}
\end{figure}

Let us start by recalling the concept of a PR-box (due to Popescu-Rohrlich~[\citenum{PoRo94}]) which is a non-signalling behaviour $\{\tilde{P}_{\scriptscriptstyle ab|xy}\}_{\scriptscriptstyle xy}$ saturating the CHSH expressions. To be more precise, only one of the expressions Eqs.~(\ref{S1})-(\ref{S4}) can reach its maximal/minimal value $\tilde{S}_{\scriptscriptstyle i}=\pm\,4$ and then the rest must be equal to zero (cf. \textbf{Remark~\ref{Lemma-inequality}}). There are eight possible PR-boxes out of which we select just two cases for illustration
\begin{eqnarray}\label{PR-box-1}
\tilde{P}^{\scriptscriptstyle\,\textit{(1)}}_{\scriptscriptstyle ab|xy}\ =\ \left\{\begin{array}{ll}\nicefrac{1}{2}&\text{\quad for\ $a\oplus b=x\cdot y$}\,,\\0&\quad\text{otherwise}\,,\end{array}\right.
\end{eqnarray}
and
\begin{eqnarray}\label{PR-box-2}
\tilde{P}^{\scriptscriptstyle\,\textit{(2)}}_{\scriptscriptstyle ab|xy}\ =\ \left\{\begin{array}{ll}0&\text{\quad for\ $a\oplus b=x\cdot y$}\,,\\\nicefrac{1}{2}&\quad\text{otherwise}\,.\end{array}\right.
\end{eqnarray}
All the remaining boxes $\{\tilde{P}^{\scriptscriptstyle\,\textit{(k)}}_{\scriptscriptstyle ab|xy}\}_{\scriptscriptstyle xy}$ for $k=3,...\,,8$ obtain by relabelling the indices $x,y$.
These two behaviours saturate the first CHSH expression Eq.~(\ref{S1}), i.e.
\begin{eqnarray}\label{PR-S}
\begin{array}{lll}
\tilde{S}^{\scriptscriptstyle\,\textit{(1)}}_{\scriptscriptstyle 1}\ =\ +\,4\,, &\quad\tilde{S}^{\scriptscriptstyle\,\textit{(1)}}_{\scriptscriptstyle i}=0&\quad\text{for $i=2,3,4$}\,,\vspace{0.1cm}\\
\tilde{S}^{\scriptscriptstyle\,\textit{(2)}}_{\scriptscriptstyle 1}\ =\ -\,4\,, &\quad\tilde{S}^{\scriptscriptstyle\,\textit{(2)}}_{\scriptscriptstyle i}=0&\quad\text{for $i=2,3,4$}\,,
\end{array}
\end{eqnarray}
and similarly the remaining six boxes saturate the other three CHSH expressions Eqs.~(\ref{S2})-(\ref{S4}) (in particular we have $\tilde{S}^{\scriptscriptstyle\,\textit{(k)}}_{\scriptscriptstyle 1}=0$ for $k=3,...\,,8$). In what follows, PR-boxes will be used because of their role in the description of the non-signalling polytope in a two-setting and two-outcome Bell scenario~[\citenum{Sc19}].

For reference let us quote \textbf{Lemma~\ref{Lemma-decomposition}} from the main manuscript that we will shortly prove; see~[\citenum{Pi03}] for a related result. See Fig.~\ref{Fig_Decomposition} for illustration.
\setcounter{lemma}{2}
\begin{lemma}\label{Lemma-decomposition-Appendix}
Each non-signalling behaviour $\{P_{\scriptscriptstyle ab|xy}\}_{\scriptscriptstyle xy}$ with binary settings $x,y\in\mathfrak{M}=\{0,\,1\}$ can be decomposed as a convex mixture of a local behaviour $\{\bar{P}_{\scriptscriptstyle ab|xy}\}_{\scriptscriptstyle xy}$ and a PR-box $\{\tilde{P}_{\scriptscriptstyle ab|xy}\}_{\scriptscriptstyle xy}$ in the form
\begin{eqnarray}\label{decomposition}
P_{\scriptscriptstyle ab|xy}\ =\ p\cdot \bar{P}_{\scriptscriptstyle ab|xy}+(1-p)\cdot \tilde{P}_{\scriptscriptstyle ab|xy}\,,
\end{eqnarray}
with $p=\tfrac{1}{2}(4-S_{\scriptscriptstyle max})$ for all $x,y\in\{0,1\}$.
\newline
[In the case $S_{\scriptscriptstyle 1}>2$, the PR-box is that of Eq.~(\ref{PR-box-1})].
\end{lemma}
\noindent Here, local behaviour means existence of an LHV+FHV model of $\{\bar{P}_{\scriptscriptstyle ab|xy}\}_{\scriptscriptstyle xy}$ and $S_{\scriptscriptstyle max}=\max\,\{|S_{\scriptscriptstyle i}|:i=1,...\,,4\}$.

\begin{proof}
Let us consider a given non-signalling behaviour $\{P_{\scriptscriptstyle ab|xy}\}_{\scriptscriptstyle xy}$ such that $S_{\scriptscriptstyle 1}>2$, and define the following four distributions
\begin{eqnarray}\label{decomposition-ansatz}
\bar{P}_{\scriptscriptstyle ab|xy}\ :=\ \tfrac{1}{p}\ P_{\scriptscriptstyle ab|xy}\ -\ \tfrac{1-p}{p}\ \tilde{P}^{\scriptscriptstyle\,\textit{(1)}}_{\scriptscriptstyle ab|xy}\,,
\end{eqnarray}
for $x,y\in\{0,1\}$, where $\{\bar{P}^{\scriptscriptstyle\,\textit{(1)}}_{\scriptscriptstyle ab|xy}\}_{\scriptscriptstyle xy}$ is the PR-box defined in Eq.~(\ref{PR-box-1}). In order to prove \textbf{Lemma~\ref{Lemma-decomposition-Appendix}} it is enough to show that
\begin{itemize}
\item[\textit{(a)}]{$\bar{P}_{\scriptscriptstyle ab|xy}$ are well-defined probability distributions\\for each $x,y\in\{0,1\}$, and }
\item[\textit{(b)}]{$\{\bar{P}_{\scriptscriptstyle ab|xy}\}_{\scriptscriptstyle xy}$ defines a local behaviour,}
\end{itemize}
for the choice $p=\tfrac{1}{2}(4-S_{\scriptscriptstyle 1})$.

\vspace{0.2cm}

\underline{\textit{Ad.~(a).}} The normalisation condition $\sum_{\scriptscriptstyle ab}\bar{P}_{\scriptscriptstyle ab|xy}=1$ is trivially satisfied, since $P_{\scriptscriptstyle ab|xy}$ and $\tilde{P}^{\scriptscriptstyle\,\textit{(1)}}_{\scriptscriptstyle ab|xy}$ are normalised. It only remains to check that $\bar{P}_{\scriptscriptstyle ab|xy}\geqslant0$.

It is known that the non-signalling polytope for the two-setting and two-outcome scenario is characterised by 24\,=\,16\,+\,8 extremal points~[\citenum{Sc19}] (cf. Fig.~\ref{Fig_Decomposition}):
\begin{eqnarray}\nonumber
\begin{array}{lcl}\{\bar{P}^{\scriptscriptstyle\,\textit{(j)}}_{\scriptscriptstyle ab|xy}\}_{\scriptscriptstyle xy}&\text{--}&\text{sixteen local deterministic behaviours}\\&&\text{for $j=1,...\,,16$\,,}\vspace{0.2cm}\\
\{\tilde{P}^{\scriptscriptstyle\,\textit{(k)}}_{\scriptscriptstyle ab|xy}\}_{\scriptscriptstyle xy}&\text{--}&\text{eight PR-boxes for $k=1,...\,,8$\,.}
\end{array}
\end{eqnarray}
This means that we can decompose $\{P_{\scriptscriptstyle ab|xy}\}_{\scriptscriptstyle xy}$ as a convex combination
\begin{eqnarray}\label{convex-decomposition}
P_{\scriptscriptstyle ab|xy}\ =\ \sum_{\scriptscriptstyle j=1}^{\scriptscriptstyle 16}\ p_{\scriptscriptstyle j}\cdot\bar{P}^{\scriptscriptstyle\,\textit{(j)}}_{\scriptscriptstyle ab|xy}\ +\ \sum_{\scriptscriptstyle k=1}^{\scriptscriptstyle 8}\ q_{\scriptscriptstyle k}\cdot\tilde{P}^{\scriptscriptstyle\,\textit{(k)}}_{\scriptscriptstyle ab|xy}\,
\end{eqnarray}
with proper normalisation $\sum_{\scriptscriptstyle j=1}^{\scriptscriptstyle 16}\ p_{\scriptscriptstyle j} + \sum_{\scriptscriptstyle k=1}^{\scriptscriptstyle 8}\ q_{\scriptscriptstyle k}=1$ and $p_{\scriptscriptstyle j},q_{\scriptscriptstyle k}\geqslant0$ for all $j,k$. We can use Eq.~(\ref{convex-decomposition}) to calculate the first CHSH expression Eq.~(\ref{S1}), which by linearity gives
\begin{eqnarray}
S_{\scriptscriptstyle1}\ =\ \sum_{\scriptscriptstyle j=1}^{\scriptscriptstyle 16}\ p_{\scriptscriptstyle j}\cdot \bar{S}^{\scriptscriptstyle\,\textit{(j)}}_{\scriptscriptstyle 1}\ +\ \sum_{\scriptscriptstyle k=1}^{\scriptscriptstyle 8}\ q_{\scriptscriptstyle k}\cdot \tilde{S}^{\scriptscriptstyle\,\textit{(k)}}_{\scriptscriptstyle1}\,,
\end{eqnarray}
where $\bar{S}^{\scriptscriptstyle\,\textit{(j)}}_{\scriptscriptstyle 1}$ and $\tilde{S}^{\scriptscriptstyle\,\textit{(k)}}_{\scriptscriptstyle 1}$ correspond to $\{\bar{P}^{\scriptscriptstyle\,\textit{(j)}}_{\scriptscriptstyle ab|xy}\}_{\scriptscriptstyle xy}$ and $\{\tilde{P}^{\scriptscriptstyle\,\textit{(k)}}_{\scriptscriptstyle ab|xy}\}_{\scriptscriptstyle xy}$ respectively. For this particular choice of distributions, we have $|\bar{S}^{\scriptscriptstyle\,\textit{(j)}}_{\scriptscriptstyle 1}|\leqslant2$ (for local behaviours) and Eq.~(\ref{PR-S}) (for PR-boxes), from which the following bound is obtained
\begin{eqnarray}\nonumber
S_{\scriptscriptstyle1}&\!\!=\!\!&\sum_{\scriptscriptstyle j=1}^{\scriptscriptstyle 16}\ p_{\scriptscriptstyle j}\cdot \bar{S}^{\scriptscriptstyle\,\textit{(j)}}_{\scriptscriptstyle 1}\ +\ 4\cdot q_{\scriptscriptstyle 1}\ -\ 4\cdot q_{\scriptscriptstyle 2}\\
&\!\!\leqslant\!\!&2\cdot\sum_{\scriptscriptstyle j=1}^{\scriptscriptstyle 16}\ p_{\scriptscriptstyle j}\ +\ 4\cdot q_{\scriptscriptstyle 1}\ \leqslant\ 2\,+\,2\cdot q_{\scriptscriptstyle 1}\,.
\end{eqnarray}
Note that in the last inequality we used the normalisation $\sum_{\scriptscriptstyle j=1}^{\scriptscriptstyle 16}\ p_{\scriptscriptstyle j} + \sum_{\scriptscriptstyle k=1}^{\scriptscriptstyle 8}\ q_{\scriptscriptstyle k}=1$. Thus, in the decomposition of Eq.~(\ref{convex-decomposition}) the coefficient $q_{\scriptscriptstyle 1}\geqslant \tfrac{1}{2}S_{\scriptscriptstyle1}-1$. A quick comparison of Eqs.~(\ref{decomposition-ansatz}) and (\ref{convex-decomposition}) reveals that such a defined distribution is well-defined, i.e.,  $\bar{P}_{\scriptscriptstyle ab|xy}\geqslant0$, as long as $q_1\geqslant1-p$. This is true for $p\geqslant \tfrac{1}{2}(4-S_{\scriptscriptstyle1})$.

\vspace{0.2cm}

\underline{\textit{Ad.~(b).}} Since both sets of behaviours $\{P_{\scriptscriptstyle ab|xy}\}_{\scriptscriptstyle xy}$ and the PR-box $\{\tilde{P}^{\scriptscriptstyle\,\textit{(1)}}_{\scriptscriptstyle ab|xy}\}_{\scriptscriptstyle xy}$ are non-signalling, then $\{\bar{P}_{\scriptscriptstyle ab|xy}\}_{\scriptscriptstyle xy}$ defined in Eq.~(\ref{decomposition-ansatz}) is non-signalling too. Therefore, by virtue of Fine's theorem~[\citenum{Fi82a,Ha14b}], it suffices to check that all CHSH inequalities are satisfied $|\bar{S}_{\scriptscriptstyle i}|\leqslant2$.

By linearity of the CHSH expresions, from Eq.~(\ref{decomposition-ansatz}) we get $\bar{S}_{\scriptscriptstyle i}=\tfrac{1}{p}\,S_{\scriptscriptstyle i}-\tfrac{1-p}{p}\,\tilde{S}^{\scriptscriptstyle\,\textit{(1)}}_{\scriptscriptstyle i}
$. It follows that
\begin{eqnarray}\nonumber
|\bar{S}_{\scriptscriptstyle i}|&\!\!=\!\!&\tfrac{|S_{\scriptscriptstyle i}\,-\,(1-p)\,\tilde{S}^{\scriptscriptstyle\,\textit{(1)}}_{\scriptscriptstyle i}|}{p}\\
&\!\!\stackrel{\scriptscriptstyle{(\ref{PR-S})}}{=}\!\!&\left\{\begin{array}{lll}
\tfrac{|S_{\scriptscriptstyle 1}-\,4\,(1-p)|}{p}&&\text{for\ $i=1$\,,}\\
\tfrac{|S_{\scriptscriptstyle i}|}{p}\ \ \leqslant\ \ \tfrac{4-|S_{\scriptscriptstyle 1}|}{p}&&\text{for\ $i\neq1$}\,,\end{array}
\right.\ 
\end{eqnarray}
where the last inequality for $i\neq1$ is due to \textbf{Remark~\ref{Lemma-inequality}}.
Now, it is straightforward to check that for $p=\tfrac{1}{2}(4-S_{\scriptscriptstyle 1})$ all four CHSH expressions satisfy the local bound, i.e. we have $|\bar{S}_{\scriptscriptstyle 1}|=2$ and $|\bar{S}_{\scriptscriptstyle i}|\leqslant2$ for $i\neq1$.

\end{proof}

\section*{DIRECT PROOF OF THEOREM~\ref{Theorem-ChainedBell}}
We start by considering a wide class of upper bounds on the measure free choice $\mu_{\scriptscriptstyle F}$. This result will be used to prove \textbf{Theorem~{\ref{Theorem-ChainedBell}}} with the help of correlations saturating the so called chained Bell inequalities (for a special choice of measurement settings on a Bell state).

\subsection*{Upper bounds on measure of free choice {\boldmath$\mu_{\scriptscriptstyle F}$}}

Let $x,y\in\mathfrak{M}$ with the choice of settings $\mathfrak{M}$ unspecified for the time being. Our analysis will proceed with algebraic expressions which take the following form
\begin{eqnarray}\label{S}
S\ =\ \sum_{\scriptscriptstyle x,y}\,\alpha_{\scriptscriptstyle xy}\cdot\langle ab\rangle_{\scriptscriptstyle xy}\,,
\end{eqnarray}
where $\alpha_{\scriptscriptstyle xy}\in\{0,\pm\,1\}$ are some arbitrary coefficients and $\langle ab\rangle_{\scriptscriptstyle xy}\equiv\sum_{\scriptscriptstyle a,b}\,ab\,P_{\scriptscriptstyle  ab|xy}$ is the usual correlation function. For each such expression $S$ we will look at two auxiliary quantities. The first one counts the number of non-zero terms in Eq.~(\ref{S}), i.e.,
\begin{eqnarray}\label{S-max}
S^{\scriptscriptstyle \#}\ =\ \sum_{\scriptscriptstyle x,y}\,|\alpha_{\scriptscriptstyle xy}|\,.
\end{eqnarray}
The second one is the upper bound of Eq.~(\ref{S}) with the correlation coefficients replaced by a product, i.e.,\footnote{A closer inspection of Eq.~(\ref{S-loc}) reveals the maximum is attained at one of the extremal points $a_{\scriptscriptstyle x},b_{\scriptscriptstyle y}=\pm\,1$ (since the expression $\sum_{\scriptscriptstyle x,y}\,\alpha_{\scriptscriptstyle xy}\cdot a_{\scriptscriptstyle x}b_{\scriptscriptstyle y}$ is a linear function for each variable $a_{\scriptscriptstyle x},b_{\scriptscriptstyle y}$).}
\begin{eqnarray}\label{S-loc}
S^{\scriptscriptstyle loc}\ =\ \max_{\scriptscriptstyle -1\leqslant\,a_{x}\,\leqslant1\atop -1\leqslant\,b_{y}\,\leqslant1}\ \sum_{\scriptscriptstyle x,y}\,\alpha_{\scriptscriptstyle xy}\cdot a_{\scriptscriptstyle x}b_{\scriptscriptstyle y}\,.
\end{eqnarray}
Note that both $S^{\scriptscriptstyle \#}$ and $S^{\scriptscriptstyle loc}$ depend only on the coefficients $\alpha_{\scriptscriptstyle xy}$ (and not on the correlation functions $\langle ab\rangle_{\scriptscriptstyle xy}$).

In the following, we consider an arbitrary HV model of a given set behaviour $\{P_{\scriptscriptstyle ab|xy}\}_{\scriptscriptstyle xy}$ obtained in a Bell experiment (with $x,y\in\mathfrak{M}$). The decomposition Eq.~(\ref{freedom-splitting}) entails splitting of  Eq.~(\ref{Pab|xy}) into two components
\begin{eqnarray}\label{P=Non-Free}
P_{\scriptscriptstyle ab|xy}\ =\ \sum_{\scriptscriptstyle \lambda\in\Lambda_{\scriptscriptstyle F}}P_{\scriptscriptstyle ab|xy\lambda}\cdot P_{\scriptscriptstyle \lambda}\ +\sum_{\scriptscriptstyle \lambda\in\Lambda_{\scriptscriptstyle NF}}P_{\scriptscriptstyle ab|xy\lambda}\cdot P_{\scriptscriptstyle \lambda|xy}\,,
\end{eqnarray}
which describe fundamentally different situations \textit{with} and \textit{without} freedom of choice, i.e. corresponding to the hidden variable $\lambda$ being respectively in the \textit{free domain} ($\Lambda_{\scriptscriptstyle F}$) or \textit{non-free domain} ($\Lambda_{\scriptscriptstyle NF}$). This can be used to write the correlation coefficients in the form
\begin{eqnarray}
\langle ab\rangle_{\scriptscriptstyle xy}
\ =\ \sum_{\scriptscriptstyle\lambda\in\Lambda_{\scriptscriptstyle F}}\langle ab\rangle_{\scriptscriptstyle xy\lambda}\cdot P_{\scriptscriptstyle\lambda}\ +\!\sum_{\scriptscriptstyle\lambda\in\Lambda_{\scriptscriptstyle{NF}}}\langle ab\rangle_{\scriptscriptstyle xy\lambda}\cdot P_{\scriptscriptstyle\lambda|xy}\,,
\end{eqnarray}
where $\langle ab\rangle_{\scriptscriptstyle xy\lambda}\equiv\sum_{\scriptscriptstyle a,b}\,ab\,P_{\scriptscriptstyle  ab|xy\lambda}$ is the correlation coefficient for fixed $\lambda$.
Hence, the expression $S$ defined in Eq.~(\ref{S}) splits into two parts
\begin{eqnarray}\label{S=Free}
S&\!\!=\!\!&\sum_{\scriptscriptstyle\lambda\in\Lambda_{\scriptscriptstyle F}}\Big(\,\sum_{\scriptscriptstyle x,y}\,\alpha_{\scriptscriptstyle xy}\cdot\langle ab\rangle_{\scriptscriptstyle xy\lambda}\,\Big)\cdot\,P_{\scriptscriptstyle\lambda}\ +
\\\label{S=Non-Free}
&&\sum_{\scriptscriptstyle\lambda\in\Lambda_{\scriptscriptstyle NF}}\Big(\,\sum_{\scriptscriptstyle x,y}\,\alpha_{\scriptscriptstyle xy}\cdot\langle ab\rangle_{\scriptscriptstyle xy\lambda}\cdot P_{\scriptscriptstyle\lambda|xy}\,\Big)\,.
\end{eqnarray}
Note that in the first term probability $P_{\scriptscriptstyle\lambda}$ factors out due to the condition in Eq.~(\ref{freedom-splitting}) defining the \textit{free domain} $\Lambda_{\scriptscriptstyle F}$. So far, the causal structure of the HV model has not been constrained.

Now, we make the \textit{locality assumption} Eq.~(\ref{factorisation}) which entails the following factorisation $\langle ab\rangle_{\scriptscriptstyle xy\lambda}=\langle a\rangle_{\scriptscriptstyle x\lambda}\,\langle b\rangle_{\scriptscriptstyle y\lambda}$\,, 
with $\langle a\rangle_{\scriptscriptstyle x\lambda}=\sum_{\scriptscriptstyle a}\,a\,P_{\scriptscriptstyle a|x\lambda}$ and $\langle b\rangle_{\scriptscriptstyle y\lambda}=\sum_{\scriptscriptstyle b}\,b\,P_{\scriptscriptstyle b|y\lambda}$. Then the expression in Eq.~(\ref{S=Free}) is bounded by $S^{\scriptscriptstyle loc}$ of Eq.~(\ref{S-loc}), i.e.,
\begin{eqnarray}\label{estimate-1}
\text{Eq.~(\ref{S=Free})}\ \,\leqslant\ \,\sum_{\scriptscriptstyle\lambda\in\Lambda_{\scriptscriptstyle F}}S^{\scriptscriptstyle loc}\cdot\,P_{\scriptscriptstyle\lambda}\ =\ S^{\scriptscriptstyle loc}\cdot\sum_{\scriptscriptstyle\lambda\in\Lambda_{\scriptscriptstyle F}} P_{\scriptscriptstyle\lambda}\,.
\end{eqnarray}
For the expression in  Eq.~(\ref{S=Non-Free}) we use $S^{\scriptscriptstyle \#}$ of Eq.~(\ref{S-max}) to make the following estimate
\begin{eqnarray}\nonumber
\text{Eq.~(\ref{S=Non-Free})}&\!\!\leqslant\!\!&\sum_{\scriptscriptstyle x,y}\,|\alpha_{\scriptscriptstyle xy}|\!\sum_{\scriptscriptstyle \lambda\in\Lambda_{\scriptscriptstyle NF}}\!P_{\scriptscriptstyle\lambda|xy}\,\stackrel{\scriptscriptstyle(\ref{star})}{=}\,\sum_{\scriptscriptstyle x,y}\,|\alpha_{\scriptscriptstyle xy}|\cdot\Big(\,1-\sum_{\scriptscriptstyle \lambda\in\Lambda_{\scriptscriptstyle F}}P_{\scriptscriptstyle\lambda}\,\Big)
\\\label{estimate-2}
&\!\!\leqslant\!\!&S^{\scriptscriptstyle \#}\cdot\Big(\,1-\sum_{\scriptscriptstyle \lambda\in\Lambda_{\scriptscriptstyle F}}P_{\scriptscriptstyle\lambda}\,\Big)\,,
\end{eqnarray}
where we use the fact that conditional probabilities $P_{\scriptscriptstyle\lambda|xy}$ are normalised and $\Lambda=\Lambda_{\scriptscriptstyle F}\cup\Lambda_{\scriptscriptstyle NF}$ (cf. Eq.~(\ref{freedom-splitting}))\,, which entails
\begin{eqnarray}\label{star}
1\ =\ \sum_{\scriptscriptstyle\lambda\in\Lambda}P_{\scriptscriptstyle\lambda|xy}\ =\ \sum_{\scriptscriptstyle\lambda\in\Lambda_{\scriptscriptstyle F}}P_{\scriptscriptstyle\lambda}\ +\sum_{\scriptscriptstyle\lambda\in\Lambda_{\scriptscriptstyle{NF}}}P_{\scriptscriptstyle\lambda|xy}
\end{eqnarray}
for each $x,y\in\mathfrak{M}$. Putting everything together, from Eqs.~(\ref{S=Free})/(\ref{S=Non-Free}) and~(\ref{estimate-1})/(\ref{estimate-2}) we get the inequality
\begin{eqnarray}\nonumber
S\ \leqslant\ S^{\scriptscriptstyle loc}\cdot \sum_{\scriptscriptstyle \lambda\in\Lambda_{\scriptscriptstyle F}}P_{\scriptscriptstyle\lambda}\,+\,S^{\scriptscriptstyle \#}\cdot(1-\sum_{\scriptscriptstyle \lambda\in\Lambda_{\scriptscriptstyle F}}P_{\scriptscriptstyle\lambda})\,,
\end{eqnarray}
which provides the following bound on the amount of free choice in a given LHV model (see \textbf{Remark~\ref{remark-definition}})
\begin{eqnarray}
\sum_{\scriptscriptstyle \lambda\in\Lambda_{\scriptscriptstyle F}}P_{\scriptscriptstyle\lambda}\ \leqslant\ \frac{S^{\scriptscriptstyle \#}-S}{S^{\scriptscriptstyle \#}-S^{\scriptscriptstyle loc}}\,.
\end{eqnarray}
Since this inequality should be satisfied by any LHV model reproducing given behaviour $\{P_{\scriptscriptstyle ab|xy}\}_{\scriptscriptstyle xy}$ in a Bell experiment (irrespective of the distribution of settings $P_{\scriptscriptstyle xy}$), then we get upper bound on the the measure of free choice defined in Eq.~(\ref{freedom-measure})
\begin{eqnarray}\label{free-choice-upper}
\mu_{\scriptscriptstyle F}\ \leqslant\ \frac{S^{\scriptscriptstyle \#}-S}{S^{\scriptscriptstyle \#}-S^{\scriptscriptstyle loc}}\,.
\end{eqnarray}
Crucially, the bound in Eq.~(\ref{free-choice-upper}) holds for any expression $S$ as defined in Eq.~(\ref{S}). Let us emphasise that the numbers $S^{\scriptscriptstyle \#}$ and $S^{\scriptscriptstyle loc}$ are determined solely by the choice of coefficients $\alpha_{\scriptscriptstyle xy}\,$, while $S$ itself depends on the observed experimental statistics $\{P_{\scriptscriptstyle ab|xy}\}_{\scriptscriptstyle xy}$ through the correlation coefficients $\langle ab\rangle_{\scriptscriptstyle xy}$\,, see Eqs.~(\ref{S})\,-\,(\ref{S-loc}).  

\subsection*{Proof of Theorem~{\ref{Theorem-ChainedBell}}}

\begin{figure}[t]
\centering
\includegraphics[width=\columnwidth]{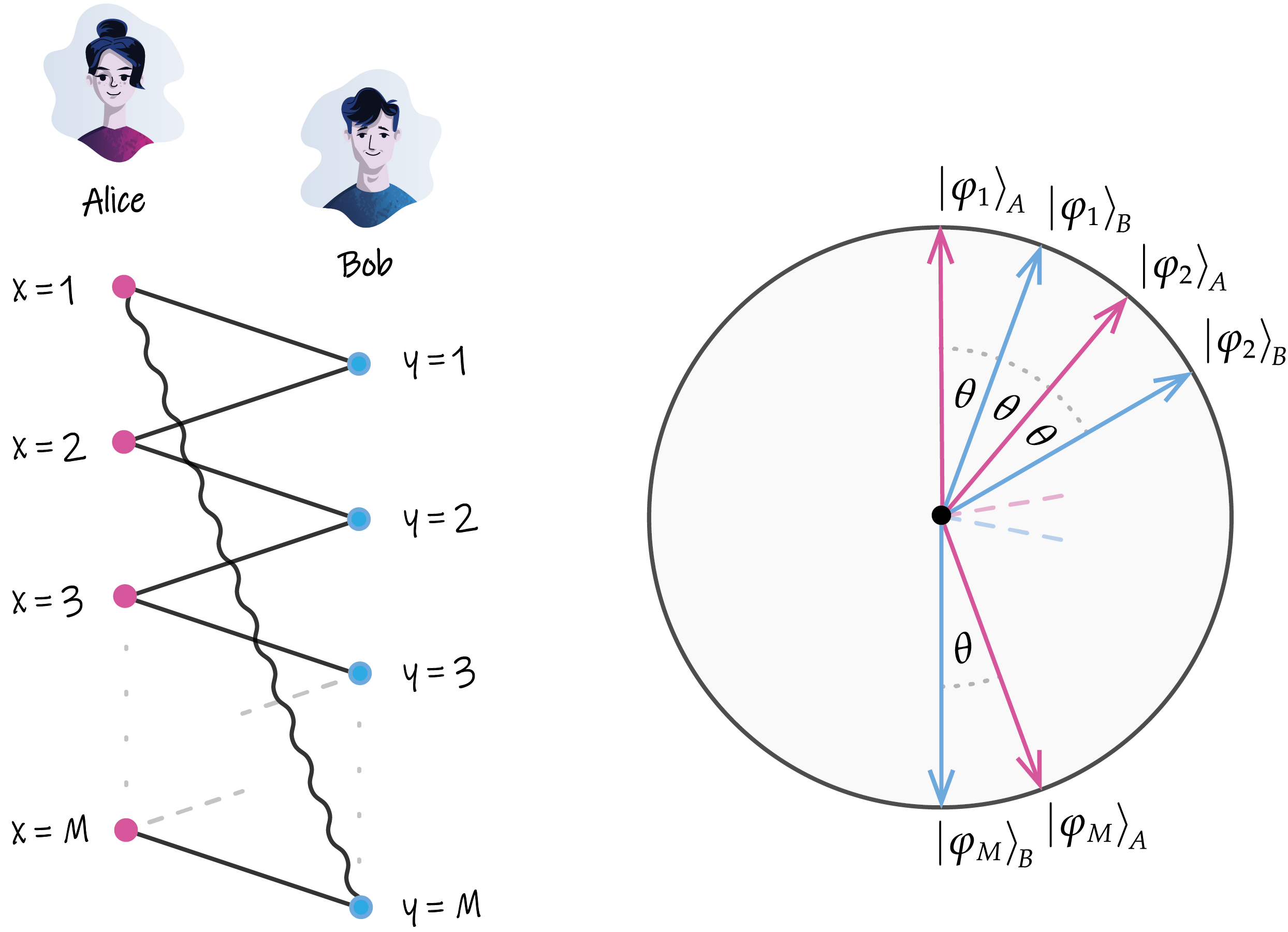}
\caption
{\label{Fig-Chained-Bell-Inequalities}{\bf\textsf{Chained Bell inequalities.}} On the left, there is a schematic illustration of the chained expression in Eq.~(\ref{Chained-Bell}) with each one of $2M$ terms $\langle ab\rangle_{\scriptscriptstyle xy}$ represented by the line $x-y$. Solid lines correspond to the '$+$' signs and the wavy line corresponds to the '$-$' sign. On the right, the $\hat{x}-\hat{z}$ section of the Bloch sphere concerns measurement directions defined in Eqs.~(\ref{chained-measurements-Alice}) and (\ref{chained-measurements-Bob}). Alice measures along the directions $\ket{\varphi_{\scriptscriptstyle x}}_{\!\scriptscriptstyle A}$ (depicted in magenta) and Bob measures along the directions $\ket{\varphi_{\scriptscriptstyle y}}_{\!\scriptscriptstyle B}$ (depicted in blue), with $x,y=1,2,\dots,M$ and $\theta=\tfrac{\pi}{2M-1}$.}
\end{figure}

Quantum correlations are stronger than the classical ones, but much weaker than what one might possibly imagine~[\citenum{PoRo94,Po14}]. For the proof of \textbf{Theorem~{\ref{Theorem-ChainedBell}}} is however enough to exclude all freedom of choice for any local account of a maximally entangled state when the number of settings $x,y\in\mathfrak{M}=\{1,2,...\,,M\}$ goes to infinity $M\rightarrow\infty$.

Let us take expression of Eq.~(\ref{S}) in a special chained form~[\citenum{Pe70,BrCa90}]
\begin{eqnarray}\nonumber
S_{\scriptscriptstyle M}&\!\!=\!\!&\langle ab\rangle_{\scriptscriptstyle 11}+\langle ab\rangle_{\scriptscriptstyle 21}+\langle ab\rangle_{\scriptscriptstyle 22}+\langle ab\rangle_{\scriptscriptstyle 32}+...\\
&&...+\langle ab\rangle_{\scriptscriptstyle MM}-\langle ab\rangle_{\scriptscriptstyle 1M}\,,
\label{Chained-Bell}
\end{eqnarray}
which follows the pattern illustrated in Fig.~\ref{Fig-Chained-Bell-Inequalities} (on the left). From simple term counting we have $S^{\scriptscriptstyle \#}_{\scriptscriptstyle M}=2\,M$. Now, if locality is imposed this leads to the so called \textit{chained Bell inequalities}: $S_{\scriptscriptstyle M}\leqslant2\,(M-1)$, see~[\citenum{Pe70,BrCa90}]. It is straightforward to see that Eq.~(\ref{S-loc}) saturates to $S^{\scriptscriptstyle loc}_{\scriptscriptstyle M}=2\,(M-1)$.\footnote{Locality entails factorisation $\expval{ab}_{\scriptscriptstyle xy}=\expval{a}_{\scriptscriptstyle x}\!\expval{b}_{\scriptscriptstyle y}=:a_{\scriptscriptstyle x}b_{\scriptscriptstyle y}$. Then $S_{\scriptscriptstyle M}=(a_{\scriptscriptstyle 1}+a_{\scriptscriptstyle 2})b_{\scriptscriptstyle 1}+(a_{\scriptscriptstyle 2}+a_{\scriptscriptstyle 3})b_{\scriptscriptstyle 2}+...+(a_{\scriptscriptstyle M}-a_{\scriptscriptstyle 1})b_{\scriptscriptstyle M}$. Suppose $a_{\scriptscriptstyle x},b_{\scriptscriptstyle y}=\pm\,1$, then out of $M$ terms taking values $0,\pm\,2$ at least one has to be equal to zero (if $a_{\scriptscriptstyle M}=a_{\scriptscriptstyle 1}$ it is the last one, otherwise it is one of the remaining ones). In such a case $S^{\scriptscriptstyle}_{\scriptscriptstyle M}\leqslant2\,(M-1)$. This extends to intermediate values $-1\leqslant a_{\scriptscriptstyle x},b_{\scriptscriptstyle y}\leqslant+1$, since $S_{\scriptscriptstyle M}$ is a linear function for each of the variables and hence the global maximum is attained at the extremal points $a_{\scriptscriptstyle x},b_{\scriptscriptstyle y}=\pm\,1$. Clearly, taking $a_{\scriptscriptstyle x}=b_{\scriptscriptstyle y}=1$ we get the maximum, i.e. $S^{\scriptscriptstyle loc}_{\scriptscriptstyle M}=2\,(M-1)$.} This means that for the expression Eq.~(\ref{Chained-Bell}) the upper bound on the amount of free choice in Eq.~(\ref{free-choice-upper}) becomes 
\begin{eqnarray}\label{free-choice-upper-chained}
\mu_{\scriptscriptstyle F}\ \leqslant\ \frac{2\,M-S_{\scriptscriptstyle M}}{2\,M-2\,(M-1)}\ =\ M-\tfrac{1}{2}S_{\scriptscriptstyle M}\,.
\end{eqnarray}

Now, we will assume that the correlations in a Bell experiment are described by the quantum statistics given by the Bell state $\ket{\Phi_{\scriptscriptstyle +}}\equiv\tfrac{1}{\sqrt{2}}(\ket{0}_{\!\scriptscriptstyle A}\ket{0}_{\!\scriptscriptstyle B}+\ket{1}_{\!\scriptscriptstyle A}\ket{1}_{\!\scriptscriptstyle B})$ and the choice of settings $x,y$ implements projective measurements in the respective bases $\big\{\!\ket{\varphi_{\scriptscriptstyle x}},\ket{\varphi_{\scriptscriptstyle x}}^{\scriptscriptstyle \perp}\!\big\}$ for Alice and $\big\{\!\ket{\varphi_{\scriptscriptstyle y}},\ket{\varphi_{\scriptscriptstyle y}}^{\scriptscriptstyle \perp}\!\big\}$ for Bob, where 
\begin{eqnarray}\label{chained-measurements-Alice}
\ket{\varphi_{\scriptscriptstyle x}}_{\!\scriptscriptstyle A}\ =\ \cos\,(x-1)\theta\,\ket{0}_{\!\scriptscriptstyle A}\,+\,\sin\,(x-1)\theta\,\ket{1}_{\!\scriptscriptstyle A},\\\label{chained-measurements-Bob}
\ket{\varphi_{\scriptscriptstyle y}}_{\!\scriptscriptstyle B}\ =\ \cos\,(y-\tfrac{1}{2})\theta\,\ket{0}_{\!\scriptscriptstyle B}\,+\,\sin\,(y-\tfrac{1}{2})\theta\,\ket{1}_{\!\scriptscriptstyle B},
\end{eqnarray}
for $x,y\in\mathfrak{M}\equiv\{1,2,\dots,M\}$ and $\theta=\tfrac{\pi}{2M-1}$. See Fig.~\ref{Fig-Chained-Bell-Inequalities} (on the right) for illustration. Then all terms in Eq.~(\ref{Chained-Bell}) become equal to $\langle ab\rangle_{\scriptscriptstyle xy}=\cos\,\theta$ except for the last one which is equal to $\langle ab\rangle_{\scriptscriptstyle 1M}=-1$, and hence we get
\begin{eqnarray}\nonumber
S_{\scriptscriptstyle M}&\!\!=\!\!&(2M-1)\cos\,\theta+1\\
&\!\!\gtrsim\!\!&(2M-1)\Big(1-\tfrac{\pi^2}{2\,(2M-1)^2}\Big)+1\,.
\end{eqnarray}
Substituted into Eq.~(\ref{free-choice-upper-chained}), this gives
\begin{eqnarray}
\mu_{\scriptscriptstyle F}\ \leqslant\ \frac{\pi^2}{4\,(2M-1)}\xymatrix{\ar[r]_{\atop M\,\rightarrow\,\infty} &}0\,.
\end{eqnarray}
This justifies our claim about freedom of choice reduced to zero in any local account of the quantum state $\ket{\Phi_{\scriptscriptstyle +}}$. Generalisation to any maximally entangled two-qubit state is straightforward.

\bibliography{CombQuant}

\makeatletter\@input{xx.tex}\makeatother